\documentclass[12pt,a4paper]{article}
\usepackage[utf8]{inputenc}
\usepackage{cmap}
\usepackage[T1]{fontenc}
\usepackage{graphicx}
\usepackage{tikz}\usetikzlibrary{calc}
\usepackage{textcomp}

\usepackage[margin=2.6cm]{geometry}
\usepackage{amsmath,amssymb,amsthm}
\usepackage[dvipsnames]{xcolor}
\usepackage[colorlinks=true,linkcolor=MidnightBlue,citecolor=MidnightBlue,urlcolor=MidnightBlue]{hyperref}

\theoremstyle{plain}
\newtheorem{theorem}{Theorem}
\newtheorem{proposition}[theorem]{Proposition}
\newtheorem{lemma}[theorem]{Lemma}
\newtheorem{corollary}[theorem]{Corollary}
\theoremstyle{definition}
\newtheorem{example}[theorem]{Example}
\newtheorem{remark}[theorem]{Remark}

\newcommand{\C}{\mathbb{C}}
\newcommand{\R}{\mathbb{R}}
\newcommand{\Z}{\mathbb{Z}}
\newcommand{\ind}{\mathbf{1}}
\newcommand{\tr}{\operatorname{Tr}}
\newcommand{\Dhat}{\hat{D}}
\newcommand{\nhat}{\hat{n}}
\newcommand{\kk}{\mathfrak{k}}
\newcommand{\tfr}{\mathfrak{t}}

\newcommand{\dist}{\operatorname{dist}}
\newcommand{\bra}[1]{\langle #1|}
\providecommand{\Kbar}{\bar K}
\providecommand{\kbar}{\bar\kappa}
\providecommand{\Dbar}{\bar D}
\newcommand{\ket}[1]{|#1\rangle}

\title{Selection rules for pinned and quasipinned\\ occupation numbers with degeneracy}
\author{Robin Reuvers\\[4pt]
\normalsize Universit\`a degli Studi Roma Tre, Dipartimento di Matematica e Fisica,\\
\normalsize L.go S. L. Murialdo 1, 00146 Roma, Italy\\
\normalsize\texttt{robin.reuvers@uniroma3.it}}
\date{}

\begin{document}
\maketitle
\begin{abstract}
When the natural occupation numbers of a fermionic state saturate a generalized Pauli constraint, the state obeys a selection rule: only those configurations of natural orbitals that saturate the constraint themselves contribute. Pinning thus singles out an active space. The selection rule was proved for non-degenerate occupation numbers; for degenerate ones it was conjectured, and proved for one saturated constraint under an assumption checked only for the constraints known at the time. Here, it is proved for every generalized Pauli constraint, with no assumption, and also in the spin-adapted setting. The rule can only fail for the ordering constraints $\lambda_j\ge\lambda_{j+1}$. Several saturated constraints are served by one basis whenever their common zero face contains a non-degenerate point. In practice, occupation numbers are quasipinned rather than pinned, and the rule survives with an explicit error bound.
\end{abstract}

\noindent\textbf{Keywords:} generalized Pauli constraints, one-body $N$-representability, pinning, quasipinning, natural orbitals, active spaces

\section{Introduction}

The Pauli principle bounds the occupation numbers of a fermionic state by one. Antisymmetry does more. For $N$ fermions in $d$ orbitals, the occupation numbers form a convex polytope, cut out by finitely many further linear inequalities known as \emph{generalized Pauli constraints} \cite{Klyachko06,AltunbulakKlyachko08}.

Two observations make these constraints physical. First, the occupation numbers of real systems---small atoms and molecules, Hubbard models---can lie on or very close to the boundary of the polytope \cite{SchillingGrossChristandl13,Schilling15,SchillingAtoms18,SchillingHubbard15}. Second, the polytope governs reduced density matrix functional theory, in which the one-body density matrix replaces the wave function as the basic variable. The constraints bound the domain of the universal functional \cite{Theophilou15}; the exact functional is built from them; and in the settings studied in \cite{SchillingSchilling19,Maciazek21}, its gradient diverges at the boundary (an ``exchange force'' that drives occupation numbers away from saturation). 

Pinning is not generic. For large systems, the generalized constraints remove only a vanishing fraction of the volume of the Pauli polytope \cite{Reuvers21}---almost everything Pauli allows remains allowed. A state that saturates a constraint is therefore special, and its structure is correspondingly informative: it is built from a restricted set of configurations, that is, a natural active space \cite{Klyachko09,SchillingGrossChristandl13,BenavidesSpringborg15,Schilling20I}. In this paper, I prove this in the case of degenerate occupation numbers, which was mostly left open by previous work (see references below), and show what becomes of it when saturation is only approximate.

\paragraph{Setting.}
The $N$-body space of fermions with one-body space $\C^d$ is $\wedge^N\C^d$. A normalized $\Psi\in\wedge^N\C^d$ has the one-body reduced density matrix
\begin{equation}\label{eq:gamma}
\gamma^\Psi:=N\tr_{2\dots N}\big[\ket{\Psi}\bra{\Psi}\big],
\end{equation}
a positive $d\times d$ matrix of trace $N$. Equivalently $\gamma^\Psi_{pq}=\langle\Psi,f^\dagger_qf_p\Psi\rangle$, where $f_p$ annihilates the $p$th vector of an orthonormal basis. The eigenvalues of $\gamma^\Psi$ in decreasing order,
\[
\lambda^\Psi=(\lambda_1,\dots,\lambda_d),\qquad \lambda_1\ge\dots\ge\lambda_d,
\]
are the \emph{natural occupation numbers}; an orthonormal eigenbasis $B=(u_1,\dots,u_d)$ with $\gamma^\Psi u_j=\lambda_ju_j$ consists of \emph{natural orbitals} \cite{Lowdin55}. The set of all occupation vectors is the polytope
\begin{equation}\label{eq:Delta}
\Delta_{d,N}:=\big\{\lambda^\Psi\ \big|\ \Psi\in\wedge^N\C^d,\ \|\Psi\|=1\big\}.
\end{equation}
That it is a polytope---the convex hull of finitely many points, hence cut out by finitely many linear inequalities---is Klyachko's theorem \cite{Klyachko06}, which is an instance of the convexity theorem for moment maps \cite{GuilleminSternberg82,Kirwan84}. Its faces are the objects of this paper. The classic case is $N=3$, $d=6$ \cite{BorlandDennis72}:
\begin{equation}\label{eq:BD}
\Delta_{6,3}=\big\{\lambda\ \big|\ 1\ge\lambda_1\ge\dots\ge\lambda_6\ge0,\ \ \lambda_j+\lambda_{7-j}=1,\ \ \lambda_4\le\lambda_5+\lambda_6\big\}.
\end{equation}
The constraints are known for all $d\le10$ \cite{AltunbulakKlyachko08}. They have the shape
\begin{equation}\label{eq:GPC}
D(\lambda):=\kappa_0+\sum_{j=1}^d\kappa_j\lambda_j\ \ge\ 0\qquad\text{for all }\lambda\in\Delta_{d,N},
\end{equation}
with real $\kappa_0,\kappa_j$. I call an affine $D$ satisfying \eqref{eq:GPC} \emph{valid}. This includes the ordering constraints $\lambda_j\ge\lambda_{j+1}$, the Pauli constraints $\lambda_1\le1$ and $\lambda_d\ge0$, and the generalized Pauli constraints proper, such as $\lambda_5+\lambda_6\ge\lambda_4$ in \eqref{eq:BD}. In electronic structure theory, one also uses \emph{spin-adapted} constraints, which describe the occupation numbers of particles with spin \cite{Liebert25} (see Section~\ref{sec:main} for the definition).

\paragraph{Pinning and the selection rule.}
A state with $D(\lambda^\Psi)=0$ for a valid $D$ is \emph{pinned} to $D$ \cite{SchillingGrossChristandl13}. Pinning has a structural consequence. Expand $\Psi$ in Slater determinants of natural orbitals,
\begin{equation}\label{eq:expansion}
\Psi=\sum_{I}c_I\ket{u_I},\qquad \ket{u_I}:=\ket{u_{i_1}\wedge\dots\wedge u_{i_N}},
\end{equation}
with $I=\{i_1<\dots<i_N\}$ running over the $N$-subsets of $\{1,\dots,d\}$, and form the one-body operator
\begin{equation}\label{eq:Dhat}
\Dhat_B:=\kappa_0+\sum_{j=1}^d\kappa_j\nhat_j,
\end{equation}
where $\nhat_j$ is the occupation number operator of $u_j$. Slater determinants are its eigenvectors, $\Dhat_B\ket{u_I}=D(\ind_I)\ket{u_I}$, where $\ind_I\in\{0,1\}^d$ is the indicator of $I$. The \emph{selection rule} says that a pinned state satisfies
\begin{equation}\label{eq:rule}
\Dhat_B\Psi=0,\qquad\text{that is,}\qquad c_I=0\ \text{ unless }\ D(\ind_I)=0, 
\end{equation}
that is, only configurations that saturate the constraint themselves contribute. The rule goes back to Klyachko \cite{Klyachko09} and to Schilling, Gross and Christandl \cite{SchillingGrossChristandl13}. It is proved when the occupation numbers are non-degenerate, $\lambda_1>\dots>\lambda_d$ \cite{Klyachko09,Walter14,Maciazek20II,Liebert25}, where it is effectively first-order perturbation theory (see Lemma~\ref{lem:nondeg}). Its consequences for active-space methods and for density matrix functional theory are developed in \cite{Schilling20I,Maciazek20II}, and its approximate form for quasipinned states in \cite{Schilling15,BenavidesSpringborg15}.

\paragraph{The degenerate case.}
Degenerate occupation numbers are common; they might for example be enforced by symmetries. Degeneracy makes the natural orbitals non-unique, and $\Dhat_B$ then depends on the choice (unless $\kappa$ is constant on the degenerate block). The selection rule becomes a statement about a basis, as the following example demonstrates.

\begin{example}\label{ex:A}
Let $u_1,\dots,u_6$ be orthonormal in $\C^6$ and, in the Borland--Dennis normal form of \cite{BorlandDennis72,SchillingGrossChristandl13},
\[
\Psi_A=\sqrt{0.5}\,\ket{u_1\wedge u_2\wedge u_3}+\sqrt{0.3}\,\ket{u_1\wedge u_4\wedge u_5}+\sqrt{0.2}\,\ket{u_2\wedge u_4\wedge u_6}.
\]
Any two of these determinants differ in two orbitals, so $\gamma^{\Psi_A}$ is diagonal by \eqref{eq:diag}, with $\lambda^{\Psi_A}=(0.8,0.7,0.5,0.5,0.3,0.2)$: pinned to $D=\lambda_5+\lambda_6-\lambda_4$ from \eqref{eq:BD}. In the basis $(u_1,\dots,u_6)$ all three configurations have $D(\ind_I)=0$ and the rule holds. But $u_3'=(u_3+u_4)/\sqrt2$, $u_4'=(u_3-u_4)/\sqrt2$ span the same eigenspace, $(u_1,u_2,u_3',u_4',u_5,u_6)$ is again an ordered eigenbasis, and re-expanding produces the configuration $\{1,2,4\}$ with $D(\ind_{\{1,2,4\}})=-1$ and coefficient $\sqrt{0.5}/\sqrt2\ne0$. In this basis, the rule fails.
\end{example}

A correct statement must therefore be existential: \emph{some} ordered basis of natural orbitals obeys the rule. Such a rule was conjectured in \cite{Schilling15}; in this form it is \cite[Conjecture~9]{Schilling20I}, due to Schilling, Benavides-Riveros, Lopes, Maci\k{a}\.zek and Sawicki. They prove it as \cite[Theorem~10]{Schilling20I} for a state saturating exactly one constraint, conditional on a combinatorial hypothesis \cite[Assumption~13]{Maciazek20II} verified for all constraints known at the time ($N\le5$, $d\le11$) but open in general. The recent work \cite{Liebert25} on spin-adapted constraints proves the non-degenerate rule and excludes degeneracies explicitly.

\paragraph{Contributions.}
I prove the selection rule for each generalized Pauli constraint, without that assumption, and also for several constraints in one basis when their common zero face contains a non-degenerate point. Example~\ref{ex:D} shows that the latter hypothesis cannot be dropped. The clean statement is Theorem~\ref{thm:practical}: every state pinned to a generalized Pauli constraint has an ordered basis of natural orbitals obeying the selection rule---degenerate or not, and for spin-adapted constraints as well. The Pauli constraints $\lambda_1\le1$ and $\lambda_d\ge0$ are not always facets and are therefore not covered by the theorem, but for them the rule is elementary (see the fifth remark after Theorem~\ref{thm:main}). Only the ordering constraints are excluded, and they must be: for example, a state pinned to $\lambda_2\ge\lambda_3$ in $\Delta_{6,3}$ violates the rule in every ordered basis of natural orbitals (see Example~\ref{ex:B}). Behind Theorem~\ref{thm:practical} is Theorem~\ref{thm:main}, valid for any inequality whose zero face contains a non-degenerate occupation vector, which also treats several saturated constraints in one and the same basis.

Exact pinning appears to be rare in practice, but the occupation numbers of atoms and molecules might be \emph{quasipinned}, with $D(\lambda^\Psi)$ small but not zero \cite{Schilling15,SchillingAtoms18}, although much of the quasipinning observed so far is a consequence of spin symmetry \cite{Liebert25}. Theorem~\ref{thm:quasi} says what the selection rule then becomes. Degeneracy is an obstruction here: Example~\ref{ex:C} exhibits states with $D(\lambda^\Psi)$ arbitrarily small and half their weight on forbidden configurations, so no bound in a state's own natural orbitals is possible. What does hold is that a quasipinned state is close to a pinned one and inherits its active space: in a natural orbital basis of that pinned state, the given one-body density matrix is close to diagonal and the weight on forbidden configurations is $O(f)$, with $f$ the constraint value and an explicit constant. 
For a perturbed Borland--Dennis state with $f=7\cdot10^{-4}$, whose own natural orbitals put half of its weight on forbidden configurations, this bounds the forbidden weight by $5\cdot10^{-3}$ (Example~\ref{ex:Cnum}). For Borland--Dennis the degenerate case had been settled by a direct computation, with a better constant \cite{SchillingBenavidesVrana17}; what is new here is that every constraint and every pattern of degeneracies is covered.

\paragraph{Method.}
The proof of Theorem~\ref{thm:main} is short: the non-degenerate rule, a compactness argument, and the fact from symplectic geometry that the map $\Psi\mapsto\lambda^\Psi$ is \emph{open} onto $\Delta_{d,N}$ \cite{Knop02,Sjamaar98,LMTW98}. Nothing in it is specific to fermions, and so Appendix~\ref{app:general} proves it for general compact Hamiltonian group actions. The spin-adapted case, the joint spin--orbital constraints of
Altunbulak and Klyachko \cite{AltunbulakKlyachko08}, and distinguishable particles come out at no extra cost. Theorem~\ref{thm:practical} needs in addition to know exactly when $\Delta_{d,N}$ and $\Delta^S_{d,N}$ are full-dimensional, which is discussed in Appendices~\ref{app:dim} and~\ref{app:spin}. The quantitative part of Theorem~\ref{thm:quasi} rests on the identity \eqref{eq:average} and on a lower bound for the slope of $\Psi\mapsto D(\lambda^\Psi)$ at degenerate states, which Ekeland's variational principle \cite{Ekeland74}  turns into a distance to the pinned states with an explicit constant (Appendix~\ref{app:exponent}). The ingredients there are first-order perturbation theory, majorization and Ekeland's principle; no symplectic geometry enters.

\section{Results}\label{sec:main}
\subsection*{Notation}
Throughout, $\Psi\in\wedge^N\C^d$ is normalized, and $\gamma^\Psi$, $\lambda^\Psi$, $\Delta_{d,N}$ are as in \eqref{eq:gamma}--\eqref{eq:Delta}.
\begin{itemize}
\item An affine $D$ is \emph{valid} if $D\ge0$ on $\Delta_{d,N}$. Its \emph{zero face} is $F_D:=\{\lambda\in\Delta_{d,N}\,|\,D(\lambda)=0\}$. When $F_D\ne\emptyset$, it is the set where $D$ attains its minimum over the polytope, since $D\ge0$ there, and hence it is a face. Every face of $\Delta_{d,N}$ is the zero face of some valid $D$. $\Psi$ is \emph{pinned} to $D$ if $\lambda^\Psi\in F_D$. A point $\lambda\in\Delta_{d,N}$ is \emph{non-degenerate} if $\lambda_1>\dots>\lambda_d$. 
\item An \emph{ordered natural orbital basis} of $\Psi$ is an orthonormal basis $B=(u_1,\dots,u_d)$ with $\gamma^\Psi u_j=\lambda^\Psi_ju_j$ for all $j$. It is unique up to phases when $\lambda^\Psi$ is non-degenerate; otherwise any orthonormal basis of each eigenspace, in any order, will do.
\item The \emph{second quantization} of a $d\times d$ matrix $h$ acts by $\Gamma(h)\ket{x_1\wedge\dots\wedge x_N}=\sum_k\ket{x_1\wedge\dots\wedge hx_k\wedge\dots\wedge x_N}$, and by \eqref{eq:gamma},
\begin{equation}\label{eq:Tr}
\langle\Psi,\Gamma(h)\Psi\rangle=\tr\big[h\gamma^\Psi\big].
\end{equation}
Thus, $\nhat_j=\Gamma(\ket{u_j}\bra{u_j})$ and $\Dhat_B=\kappa_0+\Gamma(K_B)$ with $K_B:=\sum_j\kappa_j\ket{u_j}\bra{u_j}$. For unitary $u$, $u^{\wedge N}$ denotes $u^{\otimes N}$ restricted to $\wedge^N\C^d$; then $u^{\wedge N}\Gamma(h)(u^{\wedge N})^{-1}=\Gamma(uhu^{-1})$.
\item A \emph{facet} is a face of codimension one in $\Delta_{d,N}$. The facets are the zero faces of the constraints listed after \eqref{eq:GPC}, which are the inequalities one actually writes down. A facet is an \emph{ordering facet} if, restricted to the polytope, its inequality is a positive multiple of $\lambda_j\ge\lambda_{j+1}$ for some $j$. Every other facet is a Pauli or a generalized Pauli constraint.
\item Two consequences of \eqref{eq:gamma} will be used repeatedly. In any orthonormal basis, the diagonal of $\gamma^\Psi$ is $\langle u_j,\gamma^\Psi u_j\rangle=\sum_{I\ni j}|c_I|^2$, while the off-diagonal element $\langle u_j,\gamma^\Psi u_k\rangle$ receives contributions only from pairs $I$, $I'=(I\setminus\{j\})\cup\{k\}$ of configurations differing in exactly one orbital. Hence, first, if no two configurations in \eqref{eq:expansion} differ in exactly one orbital, then $\gamma^\Psi$ is diagonal in $B$, with
\begin{equation}\label{eq:diag}
\lambda^\Psi=\text{decreasing rearrangement of }\Big(\sum_{I\ni j}|c_I|^2\Big)_{j} .
\end{equation}
Second, in an ordered natural orbital basis we have $\lambda_j=\sum_{I\ni j}|c_I|^2$ exactly, so
\begin{equation}\label{eq:average}
D(\lambda^\Psi)=\sum_I|c_I|^2\,D(\ind_I),
\end{equation}
i.e. the constraint is the average of $D$ over the configurations, weighted by their probabilities.
\end{itemize}

\paragraph{Spin-adapted constraints.}
In electronic structure theory, the one-body space is $\C^d\otimes\C^2$, and one works in the eigenspace $\mathcal{H}_N^{S,M}\subseteq\wedge^N(\C^d\otimes\C^2)$ of total spin $S$ and magnetic quantum number $M$. By Schur--Weyl duality, $\wedge^N(\C^d\otimes\C^2)\cong\bigoplus_S\mathcal{V}^S_{d,N}\otimes\C^{2S+1}$ as a representation of $U(d)\times SU(2)$, where the orbital factor $\mathcal{V}^S_{d,N}$ is the irreducible representation of $U(d)$ whose Young diagram has at most two columns, of lengths $N/2+S$ and $N/2-S$; thus $\mathcal{H}^{S,M}_N=\mathcal{V}^S_{d,N}\otimes\ket{S,M}$. The relevant one-body object is the spatial density matrix $\gamma_l^\Psi:=\tr_{\C^2}\gamma^\Psi$, positive of trace $N$ with eigenvalues $\lambda^\Psi_l\in[0,2]^d$, ordered decreasingly. The polytope $\Delta^{S}_{d,N}:=\{\lambda^\Psi_l\,|\,\Psi\in\mathcal{H}_N^{S,M},\ \|\Psi\|=1\}$ is convex and independent of $M$ \cite{AltunbulakKlyachko08,Liebert25}; its facets are the spin-adapted generalized Pauli constraints. With $u_1,\dots,u_d$ an ordered eigenbasis of $\gamma_l^\Psi$ and $\nhat_j:=\nhat_{j\uparrow}+\nhat_{j\downarrow}$, define $\Dhat_B$ by \eqref{eq:Dhat}. Its eigenvectors are the configuration state functions $\ket{I}$ of \cite{Liebert25}, which span $\mathcal{H}^{S,M}_N$ and satisfy $\Dhat_B\ket{I}=D(w_I)\ket{I}$ for the spatial occupation vector $w_I\in\{0,1,2\}^d$. All terminology transfers verbatim. For $S=0$ there are no spin-adapted constraints beyond $\lambda_1\le2$ \cite[App.~C]{Liebert25}.

\subsection{The selection rule for degenerate occupation numbers}

\begin{theorem}\label{thm:practical}
Let $D$ define a facet of $\Delta_{d,N}$ or of $\Delta^S_{d,N}$ that is not an ordering facet. Then every $\Psi$ pinned to $D$ has an ordered natural orbital basis $B$ with $\Dhat_B\Psi=0$.
\end{theorem}

Theorem~\ref{thm:practical} is proved in Section~\ref{sec:proof}, using two facts about the polytopes established in Appendices~\ref{app:dim} and~\ref{app:spin}. Figure~\ref{fig:bd} shows the statement, and two of the examples, for the Borland--Dennis polytope. In terms of \eqref{eq:rule}: some choice of natural orbitals makes only configurations with $D(\ind_I)=0$ contribute. This is the conjecture of \cite{Schilling20I} without Assumption~13 of \cite{Maciazek20II}, and with the spin-adapted constraints of \cite{Liebert25} included. Ordering facets cannot be added: Example~\ref{ex:B} gives a state pinned to $\lambda_2\ge\lambda_3$ for which the rule fails in every ordered natural orbital basis. The conjecture of \cite{Schilling20I} is stated for generalized Pauli constraints, which by definition exclude the ordering constraints, and Example~\ref{ex:B} shows that this exclusion is necessary rather than a convenience; in \cite[Def.~3.1]{Walter14} these facets are singled out as the trivial ones.

Theorem~\ref{thm:practical} follows from a general statement whose hypothesis is that the constraint does not by itself force a degeneracy.

\begin{theorem}\label{thm:main}
Let $D$ be valid with $F_D$ containing a non-degenerate point. Then every $\Psi$ pinned to $D$ has an ordered natural orbital basis $B$ with $\Dhat_B\Psi=0$. Moreover, $B$ can be chosen such that $\Dhat'_B\Psi=0$ simultaneously for every valid $D'$ that vanishes on $F_D$.
\end{theorem}

Theorem~\ref{thm:main} is proved in Section~\ref{sec:proof}. Five remarks.
\begin{enumerate}
\item The theorem is about faces: if $\Psi$ is pinned to $D_1,\dots,D_k$, it is pinned to
$D_1+\dots+D_k$, whose zero face is $F_{D_1}\cap\dots\cap F_{D_k}$; if that face contains a
non-degenerate point, one basis obeys all $k$ rules.
\item The choice of basis is irrelevant when $\kappa$ is non-decreasing on each degenerate block
of $\lambda^\Psi$: then \emph{every} ordered natural orbital basis obeys the rule, for every
valid $D$ and with no hypothesis on $F_D$ (Corollary~\ref{cor:mono}; the constant case is
trivial, as $\Dhat_B$ then does not depend on $B$), which proves the first case of \cite[Conjecture~1]{Schilling15}. In Examples~\ref{ex:A} and~\ref{ex:B},
$\kappa$ decreases along the degenerate block.
\item The `moreover' covers the $D'$ that $D$ implies, not every constraint $\Psi$ happens
to saturate. $\Psi_A$ of Example~\ref{ex:A} is pinned to Borland--Dennis and a basis obeys
that rule, but it also saturates $\lambda_3\ge\lambda_4$, and each of its three
configurations has exactly one particle in the eigenspace of $\lambda_3=\lambda_4$, so no
basis obeys that one (Lemma~\ref{lem:pair}).

\item The hypothesis on the face cannot be deleted (Example~\ref{ex:D}), and so \cite[Conjecture~9]{Schilling20I}, which explicitly covers several saturated constraints served by one basis, is false as stated. The hypothesis is
sufficient but not necessary, though: two facets of
$\Delta_{10,4}$ whose zero faces meet only in the vertex $(1,1,1,1,0,\dots,0)$ have no
non-degenerate point on that intersection, yet every state there is a Slater determinant
and both rules hold. Such exceptions do, however, seem to require a common face too small to carry any state
that violates the selection rule, and I expect them to be very special for that reason.
\item The Pauli constraints $\lambda_1\le1$ and $\lambda_d\ge0$ need not define facets, and
then neither theorem reaches them; but their rule is immediate and holds in \emph{every}
ordered natural orbital basis (see e.g.\ \cite[Lemma~16]{Reuvers21}).
\end{enumerate}

The following result for spin-adapted constraints is proved in Appendix~\ref{app:general}, which studies general moment maps.

\begin{corollary}[Spin-adapted constraints]\label{cor:spin}
Theorem~\ref{thm:main} holds verbatim for $\mathcal{H}_N^{S,M}$, $\gamma_l^\Psi$ and $\Delta^S_{d,N}$: if $D$ is valid on $\Delta^S_{d,N}$ and $F_D$ contains a non-degenerate point, then every $\Psi\in\mathcal{H}_N^{S,M}$ pinned to $D$ has an ordered eigenbasis $B$ of $\gamma_l^\Psi$ with $\Dhat_B\Psi=0$.
\end{corollary}

Let us now state the examples mentioned above. They rely on one observation about ordering constraints. 

\begin{lemma}[Ordering constraints]\label{lem:pair}
Let $\lambda^\Psi_j=\lambda^\Psi_{j+1}$ with two-dimensional eigenspace $W$, and let $D=\lambda_j-\lambda_{j+1}$. Write $\Psi=\Psi_0+\Psi_1+\Psi_2$ for the decomposition of $\Psi$ by the number of particles in $W$, that is $\Psi_k\in\wedge^kW\wedge\wedge^{N-k}W^\perp$. Then $\Dhat_B\Psi=0$ for some, equivalently for every, ordered natural orbital basis $B$ of $\Psi$ if and only if $\Psi_1=0$.
\end{lemma}

\begin{proof}
An ordered natural orbital basis has $(u_j,u_{j+1})=(v,v')$, an orthonormal basis of $W$, and then $\Dhat_B=\Gamma(A)$ with $A:=\ket{v}\bra{v}-\ket{v'}\bra{v'}$, so $Av=v$, $Av'=-v'$ and $A$ vanishes on $W^\perp$. Being a derivation, $\Gamma(A)$ therefore annihilates $\wedge^{N}W^\perp$, whence $\Gamma(A)\Psi_0=0$; and $\Gamma(A)(v\wedge v'\wedge\varphi)=(Av)\wedge v'\wedge\varphi+v\wedge(Av')\wedge\varphi=0$ for $\varphi\in\wedge^{N-2}W^\perp$, whence $\Gamma(A)\Psi_2=0$. Writing $\Psi_1=v\wedge\psi+v'\wedge\psi'$ with $\psi,\psi'\in\wedge^{N-1}W^\perp$ gives $\Gamma(A)\Psi_1=v\wedge\psi-v'\wedge\psi'$, two orthogonal terms. So $\Gamma(A)\Psi=0$ if and only if $\psi=\psi'=0$, that is $\Psi_1=0$, a condition independent of $(v,v')$.
\end{proof}

\begin{example}[Sharpness]\label{ex:B}
Let $u_1,\dots,u_6$ be orthonormal in $\C^6$ and
\[
\Psi_B=\sqrt{0.4}\,\ket{u_1\wedge u_2\wedge u_3}+\sqrt{0.4}\,\ket{u_1\wedge u_4\wedge u_5}
+\sqrt{0.2}\,\ket{u_2\wedge u_4\wedge u_6}.
\]
By \eqref{eq:diag}, $\gamma^{\Psi_B}$ is diagonal and $\lambda^{\Psi_B}=(0.8,0.6,0.6,0.4,0.4,0.2)$:
pinned to $\lambda_2\ge\lambda_3$ and to no other facet. The eigenspace of $\lambda_2=\lambda_3$ is $W=\mathrm{span}(u_2,u_4)$, and the first
two configurations have exactly one particle in $W$, so the rule for $D=\lambda_2-\lambda_3$
fails in every ordered natural orbital basis (Lemma~\ref{lem:pair}).
\end{example}

Figure~\ref{fig:bd} shows Examples \ref{ex:A} and \ref{ex:B} in $\Delta_{6,3}$, and indicates the idea behind the proof of the theorems.

\begin{example}[Several constraints]\label{ex:D}
Let $u_0,\dots,u_6$ be orthonormal in $\C^{10}$ and
\[
\Psi_C=u_0\wedge\big(\sqrt{0.02}\,\ket{u_1\wedge u_2\wedge u_3}
+\sqrt{0.02}\,\ket{u_1\wedge u_4\wedge u_5}
+\sqrt{0.96}\,\ket{u_2\wedge u_4\wedge u_6}\big)\in\wedge^4\C^{10},
\]
the state of Example~\ref{ex:B} with other weights and a spectator orbital, which lifts it
into $\Delta_{10,4}$, where several generalized Pauli constraints exist. By
\eqref{eq:diag},
\[
\lambda^{\Psi_C}=(1,0.98,0.98,0.96,0.04,0.02,0.02,0,0,0).
\]
It saturates
$68$ of the $124$ constraints of $\Delta_{10,4}$ (these are listed in the electronic supplement to \cite{AltunbulakKlyachko08}). These four are among them:
\begin{align*}
D_1&=2-(1,1,1,-1,1,-1,-1,-1,1,-1)\cdot\lambda,\\
D_2&=11-(4,9,4,-6,4,-6,-1,-6,-1,-1)\cdot\lambda,\\
D_3&=36-(9,29,9,-11,19,-21,-1,-21,-11,-1)\cdot\lambda,\\
D_4&=36-(9,9,19,-1,29,-21,-11,-21,-1,-11)\cdot\lambda.
\end{align*}
All four are facets, and one can check that
\[
-15D_1+6D_2-2D_3+D_4=10\,(\lambda_2-\lambda_3).
\]
The same identity holds for
the operators, so a basis with $\Dhat_{i,B}\Psi_C=0$ for all $i$ would annihilate
$(\nhat_2-\nhat_3)\Psi_C$, with $\nhat_j$ the occupation number operators of $B$. But the eigenspace of
$\lambda_2=\lambda_3$ is again $W=\mathrm{span}(u_2,u_4)$ with the first two configurations
meeting it in one orbital, which Lemma~\ref{lem:pair} forbids. Therefore, no ordered natural orbital
basis serves all four constraints. Their common zero face contains no non-degenerate point (the
identity forces $\lambda_2=\lambda_3$ on it)---Theorem~\ref{thm:main} does not apply, and
its hypothesis cannot be dropped.
\end{example}

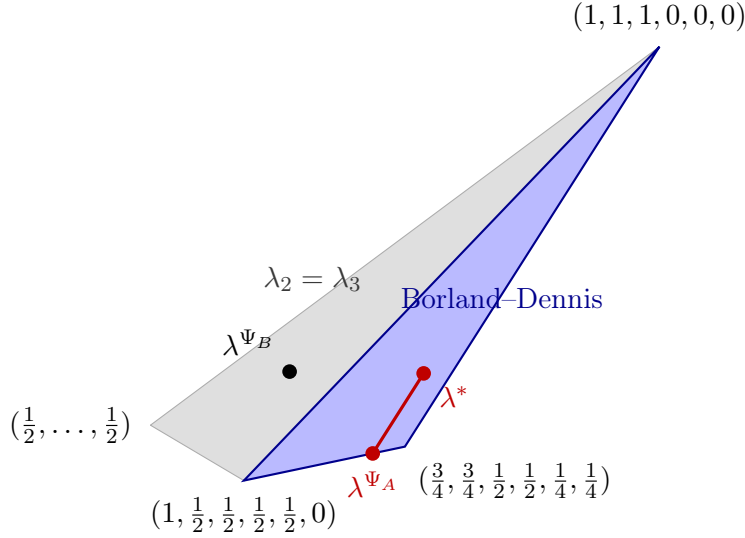
\begin{figure}[t]
\centering
\begin{tikzpicture}[scale=1.25]
\coordinate (v0) at (-2.263,-0.798);
\coordinate (v1) at (0.425,-1.025);
\coordinate (v2) at (-1.276,-1.384);
\coordinate (v3) at (3.114,3.207);
\coordinate (A) at (0.085,-1.097);
\coordinate (B) at (-0.793,-0.231);
\coordinate (S) at (0.623,-0.250);
\draw[gray!55,dashed] (v0)--(v1);
\fill[gray!25] (v0)--(v2)--(v3)--cycle;
\fill[blue!27] (v1)--(v2)--(v3)--cycle;
\draw[gray!65,thin] (v0)--(v2) (v0)--(v3);
\draw[blue!55!black,thick] (v1)--(v2)--(v3)--cycle;
\draw[red!75!black,very thick] (A)--(S);
\fill[red!75!black] (A) circle(2.2pt);
\fill[red!75!black] (S) circle(2.2pt);
\fill[black] (B) circle(2.2pt);
\node[red!75!black,anchor=north] at ($(A)+(0,-0.08)$) {\small $\lambda^{\Psi_A}$};
\node[red!75!black,anchor=north west] at ($(S)+(0.05,-0.02)$) {\small $\lambda^{*}$};
\node[anchor=south east] at ($(B)+(-0.04,0.06)$) {\small $\lambda^{\Psi_B}$};
\node[blue!55!black] at (1.45,0.55) {\small Borland--Dennis};
\node[gray!45!black] at (-0.55,0.75) {\small $\lambda_2=\lambda_3$};
\node[anchor=south] at ($(v3)+(0,0.05)$) {\small $(1,1,1,0,0,0)$};
\node[anchor=north] at ($(v2)+(0,-0.06)$) {\small $(1,\tfrac12,\tfrac12,\tfrac12,\tfrac12,0)$};
\node[anchor=north west] at ($(v1)+(0.02,-0.06)$) {\small $(\tfrac34,\tfrac34,\tfrac12,\tfrac12,\tfrac14,\tfrac14)$};
\node[anchor=east] at ($(v0)+(-0.08,0)$) {\small $(\tfrac12,\dots,\tfrac12)$};
\end{tikzpicture}
\caption{The Borland--Dennis polytope $\Delta_{6,3}$ \cite{BorlandDennis72}, a tetrahedron
with vertices as labelled. Theorem~\ref{thm:practical} applies on the blue facet
$\lambda_4=\lambda_5+\lambda_6$, while on the grey ordering facet $\lambda_2=\lambda_3$ the rule
can fail, as $\lambda^{\Psi_B}$ of Example~\ref{ex:B} shows. The two hidden facets are the
ordering facets $\lambda_1=\lambda_2$ and $\lambda_3=\lambda_4$. The ordering facets meet the blue facet in its edges, and
$\lambda^{\Psi_A}$ from Example \ref{ex:A} sits on the lower one, shared with $\lambda_3=\lambda_4$. To prove Theorem~\ref{thm:main}, we approach $\lambda^{\Psi_A}$ along a path (in red) that contains non-degenerate $\lambda^*$ in the relative interior of the blue
facet. For states corresponding to these $\lambda^*$, we know the selection rule is satisfied because they have no degenerate eigenvalues.}
\label{fig:bd}
\end{figure}

\subsection{Quasipinning}
\label{quasipinning}

As discussed in the introduction, in practice, $D(\lambda^\Psi)$ is small rather than zero. Write $f(\Psi):=D(\lambda^\Psi)$, and, for an orthonormal basis $B=(u_1,\dots,u_d)$ of $\C^d$ with $\Psi=\sum_Ic_I\ket{u_I}$, let
\[
W_B(\Psi):=\sum_{D(\ind_I)\ne0}|c_I|^2
\]
be the weight of $\Psi$ outside the active space of $B$. The selection rule \eqref{eq:rule} says $W_B(\Psi)=0$ in an ordered natural orbital basis. The question is what becomes of it when $f$ is small, but non-zero.

For non-degenerate occupation numbers the answer is known. If $f$ is small compared with the gaps between consecutive occupation numbers, a gradient flow moves $\Psi$ by at most $\sqrt{2f}$ to a pinned state, so that $W_B(\Psi)\le2f$ in the natural orbital basis $B$ of that pinned state, which is close to that of $\Psi$ \cite{SchillingBenavidesVrana17}. For Borland--Dennis, the weight in the natural orbitals of $\Psi$ itself is $O(f/(\lambda_3-\lambda_4))$ \cite[Theorem~4]{Schilling15}, and in the degenerate case it is $2f+O(f^2)$ after a rotation of the third and fourth natural orbitals \cite[Supplement]{SchillingBenavidesVrana17}, where the degenerate case is conjectured in general. Bounds in the state's own natural orbitals must degenerate as the gaps close, as demonstrated by the following example.

\begin{example}\label{ex:C}
Perturb $\Psi_A$ of Example~\ref{ex:A} to $\Psi_A+\epsilon\ket{u_1\wedge u_2\wedge u_4}$, normalized. The occupation numbers move by $O(\epsilon)$, so $f=|\epsilon|/\sqrt2+O(\epsilon^2)$ is small. But the perturbation also splits the degenerate pair $\lambda_3=\lambda_4$, and the natural orbitals it selects for the split pair converge, as $\epsilon\to0$, to $(u_3\pm u_4)/\sqrt2$: to the basis in which Example~\ref{ex:A} fails, not to the one in which it holds. In that limit, the six configurations carry weights $0.25,0.25,0.15,0.15,0.10,0.10$ and $D$-values $0,-1,1,0,1,0$. By \eqref{eq:average}, $f\to0.25(-1)+0.15(1)+0.10(1)=0$, while $W_B\to0.50$ in that basis. The constraint vanishes by \emph{cancellation} between allowed and forbidden configurations, not because the forbidden ones are absent.
\end{example}

So in the state's own natural orbitals, $W_B$ is not controlled by $f$. The reason is that a small perturbation of a degenerate pinned state fixes the natural orbitals by its own direction, with no regard for $D$. Diagonalizing a computed one-body density matrix near a degeneracy therefore returns orbitals whose selection-rule content is not controlled by the constraint value. What survives is that a quasipinned state is close to a pinned one, and the pinned one has a good basis.

\begin{theorem}[Quasipinning]\label{thm:quasi}
Let $D$ be valid with $F_D$ containing a non-degenerate point and $f\not\equiv0$, and let
$\beta_D>0$ be the combinatorial constant \eqref{eq:beta} of Appendix~\ref{app:exponent}. Every normalized $\Psi$ with
$f(\Psi)<\beta_D$ admits an orthonormal basis $B$ of $\C^d$ in which, with
$\eta:=\arcsin\sqrt{f(\Psi)/\beta_D}$,
\begin{align*}
\gamma^\Psi&=\Lambda+E\quad\text{with $\Lambda$ diagonal, its entries non-increasing, and }\ \|E\|\le\eta,\\
W_B(\Psi)&\le\eta^2,\qquad\text{and}\qquad \eta\le\sqrt{2f(\Psi)/\beta_D}\ \text{ if }
f(\Psi)\le\beta_D/2 .
\end{align*}
Here $\beta_D\le a$, the smallest positive value of $D(\ind_I)$ over configurations $I$, and
for $D$ with integer coefficients, such as every generalized Pauli constraint,
\[
\beta_D\ \ge\ \frac{2^d}{d\,(d+1)^{(d+1)/2}} .
\]
\end{theorem}

So a quasipinned state is, up to a small error, supported on the configurations that
saturate the constraint, in a basis that almost diagonalizes $\gamma^\Psi$ but need not be its own natural orbital basis---indeed need not be close to one, see Example~\ref{ex:C}. That
basis is the natural orbital basis of a pinned $\Phi$ near $\Psi$, chosen by
Theorem~\ref{thm:main} to obey the rule: with $\eta$ as in the theorem and
$\|\Psi-\Phi\|\le\eta$, the bounds $\|\gamma^\Psi-\gamma^\Phi\|\le\eta$ and
$W_B(\Psi)\le\eta^2$ are elementary, and everything rests on finding such a $\Phi$, which is
done in Theorem~\ref{thm:exponent}.

This theorem is proved by a descent, as in \cite{SchillingBenavidesVrana17}, here extended to
degenerate states. Where $f>0$, the slope of $f$ is at least $2\sqrt{f(\beta_D-f)}$. Since $f$
has no gradient at a degenerate state, it is Ekeland's variational principle \cite{Ekeland74},
rather than a flow, that turns a slope into a distance. Degeneracy costs, by the mechanism of
Example~\ref{ex:C}: the direction of the perturbation selects the natural orbitals in which the
derivative is computed, so no single basis controls the descent and one must average over them.
The constant $\beta_D$ measures how much of $f$ that averaging can cancel. The exponent
$\tfrac12$ is the best possible, already for Slater determinants (Appendix~\ref{app:exponent},
Remark~\ref{rem:sharp}). The closed-form bound of $\beta_D$ in Theorem~\ref{thm:quasi} is easy
to compute, whereas its exact value, defined in \eqref{eq:beta}, needs some computation. For a
given state the constant can improve (Proposition~\ref{prop:local}), as Example~\ref{ex:Be} shows.

\begin{example}[Example~\ref{ex:C} continued]\label{ex:Cnum}
Take $\epsilon=10^{-3}$ in Example~\ref{ex:C}. Then $f=7.1\cdot10^{-4}$, and in the natural orbital
basis of the perturbed state the forbidden weight is $W_B=0.498$. For Borland--Dennis
$\beta_D=\tfrac17$ (Appendix~\ref{app:exponent}, Remark~\ref{rem:computing}), so
Theorem~\ref{thm:quasi} gives a basis with $\|E\|\le\arcsin\sqrt{7f}=0.070$ and
$W_B\le\arcsin^2\sqrt{7f}=4.95\cdot10^{-3}$. The basis $(u_1,\dots,u_6)$ is one, with
$\|E\|=7.1\cdot10^{-4}$ and $W_B=\epsilon^2/(1+\epsilon^2)=1.0\cdot10^{-6}$. For Borland--Dennis
the bound $2f+O(f^2)=1.4\cdot10^{-3}$ of \cite[Supplement]{SchillingBenavidesVrana17} is sharper;
Theorem~\ref{thm:quasi} applies to every constraint and every pattern of degeneracies.
\end{example}

\begin{example}[Beryllium]\label{ex:Be}
For the variational ground state of beryllium in its spin-triplet sector within five $s$-type
spatial orbitals, hence ten spin-orbitals, the occupation numbers are known to fifteen digits
\cite[Table~I]{SchillingAtoms18}. Among the 124 generalized Pauli constraints of
$\Delta_{10,4}$, the most nearly saturated is $D=4-2\lambda_1-\lambda_2-\lambda_3-\lambda_5-\lambda_8$,
with $f=4.50\cdot10^{-8}$ \cite{SchillingAtoms18}. This quasipinning is a consequence of spin
symmetry \cite{Liebert25}. To fifteen digits, $\lambda_1+\lambda_2+\lambda_3+\lambda_5+\lambda_8=3$
and $\lambda_4+\lambda_6+\lambda_7+\lambda_9+\lambda_{10}=1$: these are the numbers of spin-up and
spin-down electrons of the $M=1$ state, whose natural orbitals are spin-pure. Hence $D=1-\lambda_1$
identically, the Pauli constraint $\lambda_1\le1$ is saturated equally, and in the atom's own
natural orbitals $\Dhat_B=1-\nhat_1$ on states with three spin-up electrons, so that
$W_B(\Psi)=1-\lambda_1=f$ exactly.

The state still illustrates Proposition~\ref{prop:local}. The occupation numbers form two
clusters, $\{1,\dots,4\}$ and $\{5,\dots,10\}$, each of spread about $7\cdot10^{-4}$ and separated
by $\lambda_4-\lambda_5\approx1$. Take $\rho=10^{-2}$. On the first cluster
$\kappa=(-2,-1,-1,0)$ is non-decreasing, so no interval there is admissible; on the second,
$\kappa$ drops only between positions $7$ and $8$, and the worst interval is $\{5,\dots,10\}$,
with $\beta_D(\Psi,10^{-2})=\tfrac1{11}$ (Appendix~\ref{app:exponent},
Remark~\ref{rem:computing}). Proposition~\ref{prop:local} then gives
$\dist(\Psi,R_D)\le\eta=\arcsin\sqrt{11f}=7.0\cdot10^{-4}<\rho$, hence a basis with
$\|E\|\le\eta$ and $W_B(\Psi)\le\eta^2=4.95\cdot10^{-7}$, within a factor $11$ of the weight
$f$ that the natural orbitals attain here.
\end{example}

All of this holds in the spin-adapted setting, with $\gamma_l^\Psi$ for $\gamma^\Psi$ and
Corollary~\ref{cor:spin} for Theorem~\ref{thm:main}. A spatial occupation number operator
has eigenvalues $0,1,2$, so the bound becomes $\|E\|\le2\eta$, and the displayed lower
bound for $\beta_D$ loses a factor $2^d$.

\section{Proofs}
\label{sec:proof}
The exact pinning results in Sections~\ref{labelmain} and \ref{labelproof} are self-contained, given
Theorem~\ref{thm:open} (that a small change of the occupation numbers within
$\Delta_{d,N}$ can be made by a small change of the state), which is imported from \cite{Knop02}. Appendices~\ref{app:dim} and~\ref{app:spin} supply the two facts used in
Section~\ref{labelproof}: they determine exactly when the polytopes are full-dimensional, which allows for an elegant formulation of the main theorems, but they are not essential to the proof idea. Section~\ref{label3} rests on Theorem~\ref{thm:open} and 
Theorem~\ref{thm:exponent} (that a quasipinned state is within $\arcsin\sqrt{f/\beta_D}$ of a pinned one),
proved in Appendix~\ref{app:exponent}. Appendix~\ref{app:general} supplies the moment-map
background behind Theorem~\ref{thm:open} and proves the general version of
Theorem~\ref{thm:main}, from which the spin-adapted result Corollary~\ref{cor:spin} follows.
\subsection{Proof of Theorem~\ref{thm:main}}
\label{labelmain}
I argue by approximation. Lemma~\ref{lem:nondeg} settles the states whose occupation numbers are non-degenerate, Lemma~\ref{lem:closed} shows that the conclusion survives limits, and Lemma~\ref{lem:dense} shows that under the hypothesis of the theorem every pinned state is such a limit. Together they give the theorem. Everything is elementary except the openness theorem behind Lemma~\ref{lem:dense}.

\begin{lemma}[Non-degenerate case]\label{lem:nondeg}
Let $D$ be valid and $\Psi$ pinned to $D$ with non-degenerate $\lambda^\Psi$. Then $\Dhat_B\Psi=0$ for the ordered natural orbital basis $B$ of $\Psi$.
\end{lemma}

This is first-order perturbation theory \cite{Klyachko09,Walter14,Maciazek20II,Liebert25};
in the generality of Appendix~\ref{app:general} it is \cite[Lemma~2.13]{Walter14}, proved
there as Step~1 of Theorem~\ref{thm:general}. For the proof of Theorem \ref{thm:main}, we will need the following two facts.

\begin{lemma}[Closedness]\label{lem:closed}
Let $\mathcal{D}$ be a set of affine functions on $\R^d$, and $R_\mathcal{D}$ the set of normalized $\Psi$ having an ordered natural orbital basis $B$ with $\Dhat_B\Psi=0$ for all $D\in\mathcal{D}$. Then $R_\mathcal{D}$ is closed in the unit sphere.
\end{lemma}

Here, the ordering of the basis matters: a limit of ordered natural orbital bases is again one.

\begin{proof}
Let $\Psi_n\in R_\mathcal{D}$ converge to $\Psi$, with bases $B_n=(u^{(n)}_j)$ as in the definition. Orthonormal bases form a compact set, so along a subsequence $B_n\to B=(u_j)$. By \eqref{eq:gamma}, $\gamma^{\Psi_n}\to\gamma^\Psi$, and ordered eigenvalues are continuous, so $\lambda^{\Psi_n}\to\lambda^\Psi$; passing to the limit in $\gamma^{\Psi_n}u^{(n)}_j=\lambda^{\Psi_n}_ju^{(n)}_j$ shows $B$ is an ordered natural orbital basis of $\Psi$. Finally $\Dhat_{B_n}\to\Dhat_B$ in norm, since $\Dhat_B$ depends continuously on $B$, so $\Dhat_B\Psi=\lim\Dhat_{B_n}\Psi_n=0$ for every $D\in\mathcal{D}$.
\end{proof}

\begin{theorem}[Openness \cite{Knop02,Sjamaar98}]\label{thm:open}
The map $\Psi\mapsto\lambda^\Psi$ from the unit sphere of $\wedge^N\C^d$ onto $\Delta_{d,N}$ is open: images of open sets are open in $\Delta_{d,N}$. Equivalently, for every $\Psi$ and every neighbourhood $U$ of $\Psi$, the set $\{\lambda^{\Psi'}\,|\,\Psi'\in U\}$ is a neighbourhood of $\lambda^\Psi$ in $\Delta_{d,N}$.
\end{theorem}

Openness is the converse of continuity: not only do nearby states have nearby occupation
numbers, but every nearby point of $\Delta_{d,N}$ is realized by a nearby state. This is not
elementary---the matrices $\gamma^\Psi$ are a proper subset of the positive matrices of trace
$N$, and a small change of $\lambda$ might in principle require a large change of $\Psi$.
That it does not is a theorem about moment maps: $\mathbb{P}(\wedge^N\C^d)$ is a compact
connected Hamiltonian $U(d)$-manifold, hence convex in Knop's sense
\cite[Thm.~4.2(i)]{Knop02}, and convexity in that sense implies openness
\cite[Thm.~2.2(iii)]{Knop02}. Appendix~\ref{app:general} makes the translation.

We can now prove the following lemma, and with it the theorem.

\begin{lemma}[Density]\label{lem:dense}
Let $D$ be valid with $F_D$ containing a non-degenerate point. Then every $\Psi$ pinned to $D$ is a limit of states pinned to $D$ with non-degenerate occupation numbers.
\end{lemma}

\begin{proof}
Let $\lambda:=\lambda^\Psi$ and let $\lambda^*\in F_D$ be non-degenerate. The segment $\lambda(t):=\lambda+t(\lambda^*-\lambda)$, $t\in[0,1]$, lies in $F_D$ by convexity and is non-degenerate for $t>0$, since $\lambda_j(t)-\lambda_{j+1}(t)=(1-t)(\lambda_j-\lambda_{j+1})+t(\lambda^*_j-\lambda^*_{j+1})>0$. By Theorem~\ref{thm:open}, the occupation vectors of states within $1/n$ of $\Psi$ fill a neighbourhood of $\lambda$, which contains $\lambda(t_n)$ for some $t_n>0$. Pick $\Psi_n$ within $1/n$ of $\Psi$ with $\lambda^{\Psi_n}=\lambda(t_n)$.
\end{proof}

\begin{proof}[Proof of Theorem~\ref{thm:main}]
Let $\mathcal{D}$ be the set of valid $D'$ vanishing on $F_D$. A state $\Psi'$ pinned to $D$ with non-degenerate $\lambda^{\Psi'}$ has $\lambda^{\Psi'}\in F_D$, hence is pinned to every $D'\in\mathcal{D}$, and Lemma~\ref{lem:nondeg} applied to each $D'$, with the same unique basis $B$, gives $\Psi'\in R_\mathcal{D}$. A general $\Psi$ pinned to $D$ is a limit of such $\Psi'$ by Lemma~\ref{lem:dense}, and $R_\mathcal{D}$ is closed by Lemma~\ref{lem:closed}. So $\Psi\in R_\mathcal{D}$.
\end{proof}

Corollary~\ref{cor:spin} has the same proof with spin-adapted
ingredients, and Appendix~\ref{app:general} does it once for all moment maps. The hypothesis entered once, in Lemma~\ref{lem:dense}, where it supplies the non-degenerate
$\lambda^*$ to move towards. On a face all of whose points are degenerate, an ordering face
for instance, there is nothing to approximate with. Example~\ref{ex:B} shows that this is not
an artefact of the method.

\subsection{Proof of Theorem~\ref{thm:practical}}
\label{labelproof}
Two facts about the polytopes are needed; both are proved in the appendices. Write $W_j:=\{\lambda\,|\,\lambda_j=\lambda_{j+1}\}$ for the $j$th \emph{wall}; a point is non-degenerate exactly when it lies on no wall.
\begin{itemize}
\item[(a)] If $\Delta_{d,N}$ or $\Delta^S_{d,N}$ is full-dimensional, a facet containing no non-degenerate point is an ordering facet (Lemma~\ref{lem:wall}).
\item[(b)] $\Delta_{d,N}$ is full-dimensional unless $N\le2$, $N\ge d-2$ or $(d,N)=(6,3)$ (Proposition~\ref{prop:dim}). $\Delta^S_{d,N}$ is full-dimensional unless $N\in\{2S,2d-2S\}$ and $\Delta_{d,2S}$ is not, in which case it equals $\Delta_{d,2S}$ or its particle--hole image (Corollary~\ref{cor:spinfull}).
\end{itemize}

\begin{proof}[Proof of Theorem~\ref{thm:practical}]
Let $D$ define a facet that is not an ordering facet. If $F_D$ contains a non-degenerate
point, Theorem~\ref{thm:main} or Corollary~\ref{cor:spin} applies. Otherwise the polytope is
not full-dimensional, since (a) would then make $D$ an ordering facet; so by (b) it is
$\Delta_{6,3}$, or $\Delta_{d,m}$ or its particle--hole image with $m\le2$ or $m\ge d-2$,
where $m=N$ in the fermionic case and $m=2S$ in the spin-adapted one.
Propositions~\ref{prop:ph} and~\ref{prop:polarized} realize these identifications by maps on
states, which carry ordered natural orbital bases and configurations to ordered natural orbital
bases and configurations, and preserve the equation $\Dhat_B\Psi=0$. It is therefore enough to treat $\Delta_{6,3}$ and $m\le2$. The
case $m\ge d-2$ is dual to the latter.

For $\Delta_{6,3}$, three of the four facets are ordering facets, and the zero face of the
fourth contains the non-degenerate point $(0.85,0.75,0.6,0.4,0.25,0.15)$, so
Theorem~\ref{thm:main} applies.

For $m\le1$, the polytope is a single point and has no facets. For $m=2$, the Slater decomposition $\Psi=\sum_ic_i\,u_{2i-1}\wedge u_{2i}$
\cite{Youla61,Yang62}, ordered so that $|c_1|\ge\dots\ge|c_k|$, gives
$\lambda=(t_1,t_1,\dots,t_k,t_k)$ with $t_i=|c_i|^2$ and makes $B=(u_1,\dots,u_d)$ an ordered
natural orbital basis. The polytope is a simplex in $t$, all of whose facets are ordering
facets except $t_k=0$ for even $d$. Being a facet of a simplex, it is cut out by $\lambda_d$
alone, so $D=b\,\lambda_d$ with $b>0$ modulo $\sum_j\lambda_j=2$ and
$\lambda_{2j-1}=\lambda_{2j}$. Pinning gives $c_k=0$, and every surviving configuration is a
pair $u_{2i-1}\wedge u_{2i}$, at which $\lambda_d=0$ and both relations hold. So
$\Dhat_B\Psi=0$.
\end{proof}

\subsection{Proof of Theorem~\ref{thm:quasi}}
\label{label3}
Everything here is elementary except the distance bound of Theorem~\ref{thm:exponent}, which
is proved in Appendix~\ref{app:exponent}. The rest is the observation that a state close to a
pinned one inherits its basis. 
\begin{proof}
Let $R_D$ be the set of normalized states admitting an ordered natural orbital basis $B$ with
$\Dhat_B\Psi=0$.

\emph{Step 1: $R_D$ is the set of pinned states.} In an ordered natural orbital basis
$\langle\Psi,\nhat_j\Psi\rangle=\lambda^\Psi_j$, so
$\langle\Psi,\Dhat_B\Psi\rangle=D(\lambda^\Psi)=f(\Psi)$; hence $\Dhat_B\Psi=0$ forces
$f(\Psi)=0$, and $\Psi$ is pinned. Conversely every pinned state lies in $R_D$ by
Theorem~\ref{thm:main}. So $R_D=\{f=0\}$, which is closed.

\emph{Step 2: a nearby pinned state.} By Theorem~\ref{thm:exponent},
$\dist(\Psi,R_D)\le\eta:=\arcsin\sqrt{f(\Psi)/\beta_D}$ for every normalized $\Psi$ with $f(\Psi)<\beta_D$. As $R_D$ is closed the distance is attained, say at $\Phi$; let $B$ be an
ordered natural orbital basis of $\Phi$ with $\Dhat_B\Phi=0$. Then, $\|\Psi-\Phi\|\le\eta$.

\emph{Step 3: the two bounds.} Take $\Lambda:=\gamma^\Phi$, diagonal in $B$ with non-increasing
entries, and $E:=\gamma^\Psi-\gamma^\Phi$. For a unit vector $x$, $\nhat_x:=\Gamma(\ket{x}\bra{x})$
is the occupation number operator of the orbital $x$, an orthogonal projection. The
expectations of a projection in two pure states differ by at most their trace distance, which
for pure states is at most their Euclidean distance, so by \eqref{eq:Tr}
\[
\big|\langle x,Ex\rangle\big|
=\big|\langle\Psi,\nhat_x\Psi\rangle-\langle\Phi,\nhat_x\Phi\rangle\big|\le\|\Psi-\Phi\|\le\eta,
\]
and $\|E\|=\sup_x|\langle x,Ex\rangle|\le\eta$. For the leakage, let $\Pi_B$ project onto the
span of the configurations with $D(\ind_I)\ne0$, so that $W_B(\Psi)=\|\Pi_B\Psi\|^2$. The
selection rule for $\Phi$ says exactly $\Pi_B\Phi=0$, whence
$W_B(\Psi)=\|\Pi_B(\Psi-\Phi)\|^2\le\eta^2$.
\end{proof}

\section{Conclusion}\label{sec:conclusion}

Pinning implies a selection rule also when the occupation numbers are degenerate. For every generalized Pauli constraint, the rule holds in a suitable ordered basis of natural orbitals, with no further hypothesis (Theorem~\ref{thm:practical}). For `ordinary' Pauli constraints, it holds in every such basis, for elementary reasons (see the fifth remark after Theorem~\ref{thm:main}), and it can fail only for the ordering constraints, where it actually sometimes does (Example~\ref{ex:B}). Theorem~\ref{thm:practical} thus removes Assumption~13 of \cite{Maciazek20II} from the degenerate selection rule of \cite[Theorem~10]{Schilling20I}, and Theorem~\ref{thm:main} lifts the restriction to a single saturated constraint whenever the constraints saturated do not jointly force a degeneracy. When they do, one basis need not exist (Example~\ref{ex:D}), which refutes \cite[Conjecture~9]{Schilling20I} as stated. 

The proofs use little about fermions. The one essential input is that the map from states to ordered occupation numbers is open onto its polytope, a standard property of moment maps not previously applied to this problem. The proofs therefore apply to every one-body quantum marginal problem, and Appendix~\ref{app:general} works this out for the spin-adapted constraints, for the joint constraints of Altunbulak and Klyachko on spin
and orbital occupation numbers, and for distinguishable particles.

For quasipinned states, the rule degrades gracefully, but not in the state's own natural orbitals, where Example~\ref{ex:C} rules out any bound. What holds is that a quasipinned state is close to a pinned one and inherits its active space: in a natural orbital basis of that pinned state, $\gamma^\Psi$ is within $O(\sqrt f)$ of a non-increasing diagonal matrix and the forbidden weight is $O(f)$ (Theorem~\ref{thm:quasi}). 
For the perturbed Borland--Dennis state of Example~\ref{ex:Cnum}, whose own natural orbitals put half of its weight on forbidden configurations, Theorem~\ref{thm:quasi} leaves at most $5\cdot10^{-3}$ outside the active space in a nearby basis. For the quasipinned beryllium state of \cite{SchillingAtoms18}, spin symmetry makes the constraint equal to $1-\lambda_1$, and the bound is within a factor $11$ of the weight the natural orbitals attain (Example~\ref{ex:Be}). Whether the exact constant of the method, Remark~1 of
Appendix~\ref{app:exponent}, is the smallest positive critical value of $D(\lambda^\Psi)$ in a
suitable degenerate sense is left open.

\section*{Author contribution statement}

I used Claude Opus 5 and Fable 5.1 (Anthropic) to prepare this manuscript from handwritten notes containing the proofs of Theorems \ref{thm:practical} and \ref{thm:main}. The model drafted the text, and under my guidance it found a way of improving the statement of the theorems by demonstrating the results of Appendices \ref{app:dim} and \ref{app:spin}, applied my proof method to the spin-adapted polytopes, produced the counterexample of Example \ref{ex:D}, and extended to quasipinning by proving the results of Appendix \ref{app:exponent}. I reviewed and edited as necessary. I checked all statements and proofs and take full responsibility for the content of this article.

\section*{Acknowledgements}

I thank Michael Walter for answering a question by email. This paper was written while I was visiting the Isaac Newton Institute for Mathematical Sciences, Cambridge, during the programme \emph{Mathematics of Many-Body Entanglement}, and I thank the Institute for its support and hospitality. This work was partially supported by a grant from the Simons Foundation, through a Simons Foundation Fellowship of the Institute, and by EPSRC grant EP/Z000580/1. I was partially supported by the European Research Council (ERC) under the European Union's Horizon 2020 research and innovation programme (ERC CoG UniCoSM, Grant Agreement No.~724939). I am a member of the \emph{Gruppo Nazionale per la Fisica Matematica} (GNFM) of the \emph{Istituto Nazionale di Alta Matematica Francesco Severi} (INdAM).

\appendix

\section{Dimensions of the polytopes}\label{app:dim}
This appendix proves fact (a) in the proof of Theorem~\ref{thm:practical}, and fact (b) for
$\Delta_{d,N}$. The spin-adapted half of (b) is Appendix~\ref{app:spin}. Throughout, we say that $\Delta_{d,N}$ is \emph{full-dimensional} if it has non-empty interior relative to $\{\lambda\in\R^d\,|\,\sum_j\lambda_j=N\}$, that is, dimension $d-1$. 

\subsection*{The wall lemma}

\begin{lemma}\label{lem:wall}
Let $d\ge2$ and suppose $\Delta_{d,N}$, respectively $\Delta^S_{d,N}$, is full-dimensional.
Let $D$ be valid with $F_D$ a facet containing no non-degenerate point. Then
\begin{equation}\label{eq:trivialfacet}
D(\lambda)=c\,(\lambda_j-\lambda_{j+1})+a\Big(\sum_i\lambda_i-N\Big)
\end{equation}
for some $j$, some $c>0$ and some $a\in\R$; consequently $\Dhat_B=c\,(\nhat_j-\nhat_{j+1})$ for
every ordered natural orbital basis $B$.
\end{lemma}

This is \cite[Lemma~3.2]{Walter14}, proved there for the moment polytope of any representation
of maximal dimension, which covers $\Delta^S_{d,N}$ as well (the group being $SU(d)$
\cite[Sec.~2.3]{Walter14}, $\dim\Delta_K=\dim i\tfr^*=d-1$, which is the dimension we are studying). The consequence for $\Dhat_B$ follows since the trace is fixed to $N$.

\subsection*{Dimensions}

Two constructions produce states with prescribed occupation numbers. If $\Psi_1\in\wedge^{N_1}\C^{d_1}$ and $\Psi_2\in\wedge^{N_2}\C^{d_2}$ are normalized, then $\Psi_1\wedge\Psi_2$ is a normalized state in $\wedge^{N_1+N_2}\C^{d_1+d_2}$ with $\gamma^{\Psi_1\wedge\Psi_2}=\gamma^{\Psi_1}\oplus\gamma^{\Psi_2}$. And particle--hole duality $\lambda\mapsto(1-\lambda_d,\dots,1-\lambda_1)$ maps $\Delta_{d,N}$ affinely onto $\Delta_{d,d-N}$ \cite{Ruskai70}. Applying the first construction with one factor a single orbital, either occupied or empty, gives
\begin{equation}\label{eq:embed}
\{1\}\times\Delta_{d-1,N-1}\ \subseteq\ \Delta_{d,N}\qquad\text{and}\qquad \Delta_{d-1,N}\times\{0\}\ \subseteq\ \Delta_{d,N}.
\end{equation}

\begin{lemma}[Configurations]\label{lem:config}
Let $A_1,\dots,A_m\subseteq\{1,\dots,d\}$ be configurations of $N$ particles, no two of which
differ in exactly one orbital. Then every point of the convex hull of the $\ind_{A_t}$ is the
vector of occupation numbers of a state, up to a relabelling of the orbitals. If $m=d$ and the $\ind_{A_t}$ are linearly independent in $\R^d$, $\Delta_{d,N}$ is
full-dimensional.
\end{lemma}

\begin{proof}
For $p_t\ge0$ with $\sum_tp_t=1$ and $(u_j)$ an orthonormal basis,
$\Psi_p:=\sum_t\sqrt{p_t}\ket{u_{A_t}}$ has $\gamma^{\Psi_p}$ diagonal, since off-diagonal
entries come from configurations differing in exactly one orbital, with eigenvalues
$\sum_tp_t\ind_{A_t}$ by \eqref{eq:diag}. If moreover $m=d$ and the $\ind_{A_t}$ are
independent, their hull has dimension $d-1$, so it contains a point with distinct entries.
Around it we may choose a ball, inside $\{\sum_j\lambda_j=N\}$, so small that no two entries
cross on it. There, sorting is one fixed permutation of the coordinates, so the image of the
ball is an open subset of $\Delta_{d,N}$.
\end{proof}

\begin{proposition}\label{prop:dim}
Let $d\ge2$. Then $\Delta_{d,N}$ is full-dimensional if and only if $3\le N\le d-3$ and $(d,N)\ne(6,3)$.
\end{proposition}

\begin{proof}
\emph{Exceptions.} For $N\in\{0,1,d-1,d\}$ the polytope is a point. For $N=2$ every state has
paired occupation numbers $\lambda_1=\lambda_2$ \cite{Youla61,Yang62}, so $\Delta_{d,2}$
contains no non-degenerate point, and $N=d-2$ follows by duality. For
$(d,N)=(6,3)$ the equalities in \eqref{eq:BD} reduce the dimension.

\emph{Base case $N=3$, $d\ge7$.} Let $A_t:=\{t,t+1,t+3\}\subseteq\Z_d$ be the $d$ cyclic shifts
of $A_0$. We check the hypotheses of Lemma~\ref{lem:config}. If $A_s$ and $A_t$ shared two
elements $a+s=b+t$ and $c+s=e+t$ with $a\ne c$ in $A_0$, the value $t-s$ would arise twice as a
difference $a-b$ of distinct $a,b\in A_0$; but those differences are $\pm1,\pm2,\pm3$, distinct
modulo $d$ once $d\ge7$. So $A_s$ and $A_t$ share at most one element. Stacked as rows, the $\ind_{A_t}$
form a circulant matrix with eigenvalues $1+\omega^k+\omega^{3k}$, $\omega:=e^{2\pi i/d}$, none
of which vanishes --- if $1+z+z^3=0$ with $|z|=1$, conjugating and multiplying by $z^3$ gives
$z^3+z^2+1=0$, so $z=z^2$ and $z=1$, where the value is $3$ --- so they are independent vectors in $\R^d$.

\emph{Induction on $d$.} Duality settles $N=d-3$. Let $4\le N\le d-4$. Then $(d-1,N-1)$ and
$(d-1,N)$ both satisfy $3\le N'\le d'-3$ with $d'\ge7$, so both polytopes are full-dimensional
by induction. By \eqref{eq:embed}, $\Delta_{d,N}$ contains $\{1\}\times\Delta_{d-1,N-1}$, of
dimension $d-2$ in the hyperplane $\lambda_1=1$, together with a point of
$\Delta_{d-1,N}\times\{0\}$ having $\lambda_1<1$, so their hull has dimension $d-1$.
\end{proof}

\section{The spin-adapted polytopes}\label{app:spin}
This appendix determines when $\Delta^S_{d,N}$ is full-dimensional
(Corollary~\ref{cor:spinfull}), which is the spin-adapted half of fact (b) in the proof of
Theorem~\ref{thm:practical}. Throughout, $d\ge2$, as in
Appendix~\ref{app:dim}, and $S$ is a possible total spin, so that $2S\le N$ and $N-2S$ is even, with
$\mathcal{H}^{S,M}_N\ne0$, which holds exactly when $K_0\ge0$ and $K_0+2S\le d$ for
\[
K_0:=\frac{N-2S}{2},
\]
the largest number of doubly occupied orbitals a spatial configuration of spin $S$ can have.
The answer will be that the only degenerate cases are the fully polarized one, $N=2S$, and its
particle--hole image, $N=2d-2S$. 

\begin{proposition}[Particle--hole duality]\label{prop:ph}
$\Delta^S_{d,2d-N}=\big\{(2-\lambda_d,\dots,2-\lambda_1)\ \big|\ \lambda\in\Delta^S_{d,N}\big\}$.
\end{proposition}

\begin{proof}
Let $\mathcal{C}$ be the antiunitary map on $\wedge^\bullet(\C^d\otimes\C^2)$ with
$\mathcal{C}f_p\mathcal{C}^{-1}=f^\dagger_p$ for every spin-orbital $p$. It sends $N$-particle
states to $(2d-N)$-particle states, with $\gamma^{\mathcal{C}\Psi}=\ind-\gamma^\Psi$ \cite[Eq.~(19)]{Ruskai70}, hence $\gamma_l^{\mathcal{C}\Psi}=2\cdot\ind_d-\gamma_l^\Psi$. It commutes with
$\hat S^2$ and reverses $\hat S_z$, so it maps $\mathcal{H}^{S,M}_N$ bijectively onto
$\mathcal{H}^{S,-M}_{2d-N}$. The spatial occupation numbers of $\mathcal{C}\Psi$ are therefore
$(2-\lambda_d,\dots,2-\lambda_1)$, and since $\Delta^S_{d,N}$ does not depend on $M$
\cite{AltunbulakKlyachko08,Liebert25}, the claim follows.
\end{proof}

The same map with $\C^d$ in place of $\C^d\otimes\C^2$ gives the fermionic duality at the level
of states: $\mathcal{C}$ sends $\wedge^N\C^d$ to $\wedge^{d-N}\C^d$ with
$\gamma^{\mathcal{C}\Psi}=\ind-\gamma^\Psi$, so it carries an ordered natural orbital basis
$B=(u_1,\dots,u_d)$ of $\Psi$ to the reversed basis $B'=(u_d,\dots,u_1)$ of $\mathcal{C}\Psi$.
Since $\mathcal{C}\nhat_p\mathcal{C}^{-1}=1-\nhat_p$, the constraint $D'$ with
$\kappa'_i:=-\kappa_{d+1-i}$, $\kappa'_0:=\kappa_0+\sum_j\kappa_j$ has
$\mathcal{C}\Dhat_B\mathcal{C}^{-1}=\Dhat'_{B'}$; so $\Dhat_B\Psi=0$ if and only if
$\Dhat'_{B'}\mathcal{C}\Psi=0$, and ordering facets correspond to ordering facets.

\begin{proposition}[Configurations]\label{prop:spinconfig}
Let $w^{(1)},\dots,w^{(m)}\in\{0,1,2\}^d$ be spatial occupation vectors with
$\sum_jw^{(t)}_j=N$, each with at most $K_0$ entries equal to $2$, no two differing by a vector
$e_p-e_q$. Then every point of the convex hull of the $w^{(t)}$ is the vector of spatial
occupation numbers of a state in $\mathcal{H}^{S,M}_N$, up to a relabelling of the orbitals. If
$m=d$ and the $w^{(t)}$ are linearly independent in $\R^d$, then $\Delta^S_{d,N}$ is
full-dimensional.
\end{proposition}

\begin{proof}
Let $w$ be one of the $w^{(t)}$, with $m_w$ entries equal to $2$. A configuration with these
spatial occupations has $m_w$ doubly occupied orbitals, which are singlets, and $N-2m_w$ singly
occupied ones, coupling to total spin $S$ precisely when $2S\le N-2m_w$ with $N-2m_w-2S$ even.
Both hold since $N-2m_w-2S=2(K_0-m_w)$ and $m_w\le K_0$. So there is a state
$\ket{w}\in\mathcal{H}^{S,M}_N$ with spatial occupations $w$.

The rest is the proof of Lemma~\ref{lem:config}, with the $\ket{w^{(t)}}$ in place of the Slater
determinants $\ket{u_{A_t}}$. For $i\ne j$ we have $(\gamma_l^\Psi)_{ji}=\langle\Psi,\hat
E_{ij}\Psi\rangle$, where $\hat E_{ij}$ moves one electron from spatial orbital $j$ to spatial
orbital $i$; it sends $\ket{w^{(t)}}$ to a state with spatial occupations $w^{(t)}-e_j+e_i$,
which is no $w^{(s)}$, hence orthogonal to every $\ket{w^{(s)}}$. So $\gamma_l^{\Psi_p}$ is
diagonal for $\Psi_p:=\sum_t\sqrt{p_t}\ket{w^{(t)}}$, with eigenvalues $\sum_tp_tw^{(t)}$ since
$\nhat_j\ket{w^{(t)}}=w^{(t)}_j\ket{w^{(t)}}$.
\end{proof}

\begin{proposition}\label{prop:polarized}
If $N=2S$ then $\Delta^S_{d,N}=\Delta_{d,N}$. If $N=2d-2S$ then $\Delta^S_{d,N}$ is the image of $\Delta_{d,2S}$ under $\lambda\mapsto(2-\lambda_d,\dots,2-\lambda_1)$.
\end{proposition}

\begin{proof}
If $N=2S$ the orbital factor is $\mathcal{V}^S_{d,2S}=\wedge^N\C^d$, so $\mathcal{H}^{S,M}_N=\wedge^N\C^d\otimes\ket{S,M}$ for \emph{every} $M$, and $\gamma_l$ is the fermionic one-body density matrix of the orbital factor. The case $N=2d-2S$ follows from Proposition~\ref{prop:ph}, since then $2d-N=2S$.
\end{proof}

\begin{theorem}\label{thm:spindim}
Suppose $K_0\ge1$ and $K_0+2S\le d-1$. Then $\Delta^S_{d,N}$ is full-dimensional.
\end{theorem}

\begin{proof}
It suffices to find $d$ spatial occupation vectors satisfying the hypotheses of
Proposition~\ref{prop:spinconfig}.

\emph{Reduction to $N\le d$.} Particle--hole duality sends $(N,S,d)$ to $(2d-N,S,d)$,
hence $K_0$ to $d-K_0-2S$, exchanging the two hypotheses. As one of $N$ and $2d-N$ is at most
$d$, Proposition~\ref{prop:ph} lets us assume $N\le d$; the second hypothesis is then automatic,
since $K_0+2S=\tfrac N2+S\le N-1\le d-1$ once $K_0\ge1$, and it was needed only to make this
reduction available. Note $N=2K_0+2S\ge2$.

\emph{Step 1: the configurations.} Index $\Z_d$ and let
\[
v:=(2,\underbrace{1,\dots,1}_{N-2},\underbrace{0,\dots,0}_{d-N+1})\in\{0,1,2\}^d,
\]
with cyclic shifts $w^{(t)}:=\tau^tv$, where $\tau$ shifts by one position. Each $w^{(t)}$ has
entry sum $N$ and a single entry $2$, so at most $K_0$ of them as $K_0\ge1$. It remains to check
the two hypotheses of Proposition~\ref{prop:spinconfig}.

\emph{Step 2: linear independence.} The $w^{(t)}$ are the rows of a circulant matrix with symbol
\[
f(z)=2+z+\dots+z^{N-2}=1+\frac{z^{N-1}-1}{z-1},
\]
whose eigenvalues are $f(\omega^k)$, $\omega:=e^{2\pi i/d}$. Here $f(1)=N\ne0$, and for
$\omega^k\ne1$ the equation $f(\omega^k)=0$ reads $\omega^{k(N-1)}+\omega^k=2$, a sum of two
unit-modulus terms, forcing $\omega^k=1$.

\emph{Step 3: no two shifts differ by $e_p-e_q$.} Such a difference would make $v$ invariant
under $j\mapsto j-s$ away from two positions, hence periodic with period $g=\gcd(d,s)$, which
the block shape of $v$ forbids. In detail, suppose $w^{(0)}-w^{(s)}=e_p-e_q$ with
$s\not\equiv0$. The entry $2$ of $v$ sits at $0$ and that of $\tau^sv$ at $s$, all other entries
being at most $1$, so the difference is $\ge1$ at $0$ and $\le-1$ at $s$; hence $p=0$, $q=s$,
$v_s=v_{-s}=1$, and $v_j=v_{j-s}$ for $j\notin\{0,s\}$. Put $g:=\gcd(d,s)\le d/2$. The cycles of
$j\mapsto j-s$ are the cosets of $g\Z_d$, of $d/g\ge2$ elements each. The relation $v_j=v_{j-s}$
holds everywhere except at $j=0$ and $j=s$, which lie in the cycle $g\Z_d$ and isolate $0$ in
it; so $v$ is constant on every other cycle and equals $v_s=1$ on $g\Z_d\setminus\{0\}$.

Hence $g\Z_d\setminus\{0\}\subseteq\{1,\dots,N-2\}$, so $d-g\le N-2$. The zero positions
$\{N-1,\dots,d-1\}$ are a non-empty union of cycles, so contain a coset of $g\Z_d$, which the
arc $\{0,\dots,N-2\}$ must miss; impossible unless $N-1<g$, as any $g$ consecutive integers meet
every coset. Together $\max(N,\,d-N+2)\le g\le d/2$, so $N\le d/2$ and $N\ge d/2+2$, a contradiction.
Proposition~\ref{prop:spinconfig} now applies with $m=d$.
\end{proof}

Since $\mathcal{H}^{S,M}_N\ne0$ means $K_0\ge0$ and $K_0+2S\le d$, with $K_0$ and $2S$ integers,
Theorem~\ref{thm:spindim} misses only $K_0=0$ and $K_0+2S=d$, that is $N=2S$ and $N=2d-2S$.
With Propositions~\ref{prop:polarized} and~\ref{prop:dim} we conclude the following.

\begin{corollary}\label{cor:spinfull}
$\Delta^S_{d,N}$ is full-dimensional unless $N=2S$ or $N=2d-2S$. In those two cases $\Delta^S_{d,N}$ is $\Delta_{d,2S}$ or its particle--hole image, whose dimension is given by Proposition~\ref{prop:dim}. In particular, $\Delta^S_{d,N}$ fails to be full-dimensional exactly when $N\in\{2S,2d-2S\}$ and, in addition, $2S\le2$, or $2S\ge d-2$, or $(d,2S)=(6,3)$.
\end{corollary}

Both exceptions are harmless and these cases are treated directly in the proof of Theorem~\ref{thm:practical}.

\section{Moment maps and the general selection rule}\label{app:general}
The proof in Section~\ref{sec:proof} used nothing about fermions beyond the fact that
$\Psi\mapsto\gamma^\Psi$ is a moment map, so it proves the following in general. This is also
the setting of Theorem~\ref{thm:open}, and it yields Corollary~\ref{cor:spin} and the
spin--orbital and distinguishable-particle versions for free.

\paragraph{Setting.}
Let $K$ be a compact connected Lie group with Lie algebra $\kk$ and maximal torus $T$ with Lie
algebra $\tfr$, and fix an invariant inner product to identify $\kk^*\cong\kk\supseteq\tfr$.
Every adjoint orbit meets a fixed closed Weyl chamber $\tfr^*_+$ in exactly one point, which
defines a continuous $\pi:\kk^*\to\tfr^*_+$. A point of $\tfr^*_+$ is \emph{regular} if it lies
in the interior of $\tfr^*_+$ in $\tfr^*$, equivalently if its stabilizer is $T$. The others lie
on a wall $\ker\beta$ with $\beta$ a simple root. Let $(M,\omega)$ be a compact connected
symplectic manifold with a Hamiltonian $K$-action and equivariant moment map $\mu:M\to\kk^*$, so
that for $\xi\in\kk$ the fundamental vector field $\xi_M$ is the Hamiltonian vector field of
$\mu^\xi:=\langle\mu,\xi\rangle$,
\begin{equation}\label{eq:moment}
d\mu^\xi=\omega(\xi_M,\,\cdot\,).
\end{equation}
Put $\psi:=\pi\circ\mu$. By Kirwan's theorem \cite{Kirwan84}, $\Delta:=\psi(M)$ is a convex
polytope. For $m\in M$ let $\kk_m:=\{\xi\in\kk\,|\,\xi_M(m)=0\}$ be the Lie algebra of the
stabilizer of $m$.

\paragraph{The fermionic case.}
Take $K=U(d)$ acting on $M=\mathbb{P}(\wedge^N\C^d)$ by $u\mapsto u^{\wedge N}$, with the
Fubini--Study form. Identifying $\mathfrak{u}(d)^*$ with Hermitian matrices, the moment map is
$\mu([\Psi])=\gamma^\Psi$ (up to normalization), by \eqref{eq:Tr}
\cite{Klyachko06,AltunbulakKlyachko08}, \cite[Section~2.3]{Walter14}. With $T$ the diagonal
unitaries and $\tfr^*_+$ the decreasingly ordered real diagonal matrices,
$\psi([\Psi])=\lambda^\Psi$, regular points are the non-degenerate $\lambda$, and
$\Delta=\Delta_{d,N}$. The centre of $U(d)$ acts trivially on $M$, so nothing here changes if
one works with $SU(d)$, as in Appendix~\ref{app:dim}. Finally, for Hermitian $h$ one has
$h\in\kk_{[\Psi]}$ if and only if $\Gamma(h)\Psi\in\C\Psi$, and then
$\Gamma(h)\Psi=\tr[h\gamma^\Psi]\Psi$ by \eqref{eq:Tr}.

\paragraph{Openness.}
Being compact and connected, $M$ is \emph{convex} in the sense of \cite[Thm.~4.2(i)]{Knop02},
and convexity includes the openness of $\psi$ onto $\Delta$ \cite[Thm.~2.2(iii)]{Knop02}.
Theorem~\ref{thm:open} is stated on the unit sphere rather than on $\mathbb{P}(\wedge^N\C^d)$. It follows because the quotient map $q$ to projective space is open and $\lambda^\Psi$ factors
through it, so that for $U$ open in the sphere, $\psi(q(U))=\{\lambda^{\Psi'}\mid\Psi'\in U\}$ is
open in $\Delta_{d,N}$.

\begin{theorem}\label{thm:general}
Let $\xi\in\tfr$ and $c\in\R$ be such that $\langle\alpha,\xi\rangle\ge c$ for all
$\alpha\in\Delta$, and let $F:=\{\alpha\in\Delta\,|\,\langle\alpha,\xi\rangle=c\}$. Suppose $F$
contains a regular point. Then, for every $m\in M$ with $\psi(m)\in F$, there is $k\in K$ such
that
\[
\mu(k\cdot m)\in\tfr^*_+\qquad\text{and}\qquad \xi\in\kk_{k\cdot m}.
\]
Moreover, $k$ can be chosen such that $\xi'\in\kk_{k\cdot m}$ for every $\xi'\in\tfr$ for which
$\langle\cdot,\xi'\rangle$ attains its minimum over $\Delta$ on all of $F$.
\end{theorem}

For fermions, $\xi=\mathrm{diag}(\kappa)$ and $c=-\kappa_0$, and the theorem says:
$\Psi':=k^{\wedge N}\Psi$ has diagonal, decreasingly ordered $\gamma^{\Psi'}$ and
$\Gamma(\mathrm{diag}(\kappa))\Psi'=\tr[\mathrm{diag}(\kappa)\gamma^{\Psi'}]\Psi'=-\kappa_0\Psi'$.
Translating back to $\Psi$ with the ordered natural orbital basis $u_j:=k^{-1}e_j$, this is
$\Dhat_B\Psi=0$, and Theorem~\ref{thm:main} follows. (Rescaling the Fubini--Study form rescales
$\mu$, hence $\Delta$, $\xi$ and $c$, and this changes nothing.)

\begin{proof}
\emph{Step 1: regular case} \cite[Lemma~2.13]{Walter14}. Suppose $\alpha:=\mu(m)\in\tfr^*_+$ is
regular. Then $\pi$ is smooth near $\alpha$, and $d\pi_\alpha$ is the restriction map
$\kk^*\to\tfr^*$: it is the identity on $\tfr^*$ and vanishes on the tangent space
$\{\mathrm{ad}^*_\eta\alpha\}$ of the orbit, which is the annihilator of $\tfr$. Hence
$f:=\langle\psi,\xi\rangle-c$ is smooth near $m$ with $df_m=d\mu^\xi_m$, as $\xi\in\tfr$. Since
$f\ge0$ and $f(m)=0$, $df_m=0$, so $\xi_M(m)=0$ by \eqref{eq:moment} and the non-degeneracy of
$\omega$: $\xi\in\kk_m$. The same argument applies to every $\xi'$ as in the last sentence of
the theorem. If $\mu(m)\notin\tfr^*_+$, run it at $k\cdot m$ for a $k$ with
$\mu(k\cdot m)=\psi(m)$.

\emph{Step 2: closedness.} Let $\Xi\subseteq\tfr$ be the set of all $\xi'$ as in the theorem
(note $\xi\in\Xi$) and let $R$ be the set of $m\in M$ for which there is $k\in K$ with
$\mu(k\cdot m)\in\tfr^*_+$ and $\Xi\subseteq\kk_{k\cdot m}$. If $m_n\in R$ with witnesses $k_n$
and $m_n\to m$, pass to a subsequence $k_n\to k$. Then $\mu(k_n\cdot m_n)\to\mu(k\cdot m)$,
which lies in the closed set $\tfr^*_+$, and $\xi'_M(k_n\cdot m_n)=0\to\xi'_M(k\cdot m)=0$ for
each $\xi'\in\Xi$ by continuity of the vector field $\xi'_M$. So $m\in R$, and $R$ is closed.

\emph{Step 3: density.} Let $\psi(m)=\alpha\in F$ and let $\alpha^*\in F$ be regular. For
$t\in(0,1]$ the point $\alpha(t):=\alpha+t(\alpha^*-\alpha)$ lies in $F$ by convexity and is
regular: for each simple root $\beta$,
$\langle\alpha(t),\beta\rangle=(1-t)\langle\alpha,\beta\rangle+t\langle\alpha^*,\beta\rangle>0$.
By openness, every neighbourhood of $m$ contains points $m'$ with $\psi(m')=\alpha(t)$ for some
small $t>0$. These $m'$ lie in $R$ by Step~1, so $m\in\overline{R}=R$ by Step~2.
\end{proof}

The closest statement in the literature is Sjamaar's condition at a vertex
\cite[Theorem~6.7(2)]{Sjamaar98}: if $\alpha$ is a vertex of $\Delta$ and is regular, then $T$
fixes every $m$ in the moment fibre $\mu^{-1}(\alpha)$. (At other points of the larger fibre
$\psi^{-1}(\alpha)$ a conjugate of $T$ is needed.) For fermions a $T$-fixed state is a Slater
determinant, whose occupation numbers $(1,\dots,1,0,\dots,0)$ are degenerate once $N\ge2$ or
$d-N\ge2$; so $\Delta_{d,N}$ has no non-degenerate vertex and the criterion is vacuous.
Theorem~\ref{thm:general} is the analogue for faces, where it is not. Sjamaar's argument uses
that the local cone at a vertex is proper, which fails at a face, and the hypothesis that $F$
contain a regular point replaces that input.

\begin{proof}[Proof of Corollary~\ref{cor:spin}]
The group $U(d)$ acts on $\wedge^N(\C^d\otimes\C^2)$ by $u\mapsto(u\otimes\ind_2)^{\wedge N}$.
This action commutes with the spin operators $\hat S^2$ and $\hat S_z$, which are built from
$\ind_d\otimes\sigma$, so it preserves $\mathcal{H}^{S,M}_N$, whose projectivization, with the
Fubini--Study form, is a compact connected Hamiltonian $U(d)$-manifold. By \eqref{eq:Tr} applied
to $h\otimes\ind_2$, the moment map is $[\Psi]\mapsto\gamma_l^\Psi$, so
$\psi([\Psi])=\lambda_l^\Psi$ and $\Delta=\Delta^S_{d,N}$. Theorem~\ref{thm:general} applies,
and its translation is as in the fermionic case, with
$\Gamma(\mathrm{diag}(\kappa)\otimes\ind_2)=\sum_j\kappa_j(\nhat_{j\uparrow}+\nhat_{j\downarrow})$.
\end{proof}

\paragraph{Spin and orbital occupation numbers.}
Altunbulak and Klyachko also constrain the orbital occupation numbers jointly with the spin
occupation numbers $\nu^\Psi$, the spectrum of the reduced state $\rho^\Psi_s$ on $\C^{2S+1}$, for a state of
definite $S$ but not of definite $M$ \cite[Sec.~3.1.1, Cor.~1 and
Sec.~6.1]{AltunbulakKlyachko08}. The pairs $(\lambda^\Psi_l,\nu^\Psi)$ form the moment polytope
of $U(d)\times U(2S+1)$ on $\mathbb{P}(\mathcal{V}^S_{d,N}\otimes\C^{2S+1})$, so
Theorem~\ref{thm:general} applies here too: on a face carrying a point where $\lambda_l$ and
$\nu$ are both non-degenerate, every pinned state has ordered eigenbases of $\gamma^\Psi_l$ and
$\rho^\Psi_s$ in which only the product configurations annihilated by the constraint contribute.

\begin{remark}[Distinguishable particles]
For $K=U(d_1)\times\dots\times U(d_r)$ acting on
$\mathbb{P}(\C^{d_1}\otimes\dots\otimes\C^{d_r})$, the moment map collects the $r$ one-party
marginals, the chamber is a product, and $\Delta$ is the polytope of compatible local spectra.
Theorem~\ref{thm:general} then gives a selection rule in local eigenbases for every pure state
pinned to a face of $\Delta$ that contains a point where all $r$ local spectra are
non-degenerate. For $r\ge3$ qubits the polygon inequalities \cite{HiguchiSudberySzulc03} bound the smaller
eigenvalue $\mu_i$ of each marginal by the sum of the others. On the $i$th facet, where
$\mu_i=\sum_{j\ne i}\mu_j$, taking $\mu_j$ near $0$ for $j\ne i$ forces $\mu_i$ near $0$, so all
$r$ marginals have $\mu<\tfrac12$ and are non-degenerate, and one recovers
\cite[Theorem~18]{Maciazek20II}. 
\end{remark}

\section{The constant in Theorem~\ref{thm:quasi}}\label{app:exponent}
Theorem~\ref{thm:exponent} follows once $f$ is known to decrease at a definite
rate wherever it is positive, since Ekeland's variational principle converts such a rate into a
distance to $\{f=0\}$. At a non-degenerate state the rate is read off the derivative of $f$, as
in \cite{SchillingBenavidesVrana17}. At a degenerate state $f$ has no derivative: the direction
of the perturbation selects the natural orbitals in which the rate is computed, and different
directions select different bases. The remedy is to average over those bases---the smearing of
the third subsection---and $\beta_D$, defined in the second, bounds what the averaging costs.

Throughout, $D$ is valid, $f(\Psi):=D(\lambda^\Psi)\ge0$, and $R_D:=\{f=0\}$, the same set as in
Step~1 of the proof of Theorem~\ref{thm:quasi}. Distances are chordal on the unit sphere $S$
of $\wedge^N\C^d$, with tangent space
$T_\Psi S=\{\Theta\,|\,\mathrm{Re}\langle\Psi,\Theta\rangle=0\}$. For a Hermitian $K$ write $\Dhat_K:=\kappa_0+\Gamma(K)$, so that $\Dhat_B=\Dhat_{K_B}$ for
$K_B=\sum_j\kappa_j\ket{u_j}\bra{u_j}$ (the smearing below needs $\Dhat_K$ for operators $K$
that are not of this form). A \emph{block} of a normalized $\Psi$ is a maximal interval $b$ of positions on which
$\lambda^\Psi_j$ is constant, and $\kappa_b:=(\kappa_j)_{j\in b}$. A position where
$\lambda^\Psi_j$ is simple is a block of one element.

\subsection*{Directional derivatives}
First, a technical lemma. At a degenerate state the derivative of $f$ depends on the direction through the way it splits
the blocks, so a single basis controls it only when $\kappa$ is non-decreasing along them.

\begin{lemma}\label{lem:dirderiv}
Let $\Psi$ be normalized, $\Theta\in T_\Psi S$ a unit vector,
$\Psi_s:=(\Psi+s\Theta)/\sqrt{1+s^2}$, and
$\gamma':=N\tr_{2\dots N}\big[\ket\Psi\bra\Theta+\ket\Theta\bra\Psi\big]$, the first-order change
of $\gamma^\Psi$ along $\Theta$. For a block $b$ of $\Psi$, with eigenspace $W_b$ and projection
$P_b$, let $(\mu(j))_{j\in b}$ be the eigenvalues of $P_b\gamma'P_b$ on $W_b$, labelled so as to
be non-increasing along $b$. Put $G_B:=2\big(\Dhat_B\Psi-f(\Psi)\Psi\big)$, which lies in
$T_\Psi S$. Then,
\begin{enumerate}
\item[(i)] the one-sided derivative
$f'(\Psi;\Theta):=\lim_{s\downarrow0}\big(f(\Psi_s)-f(\Psi)\big)/s$ exists and equals
$\sum_j\kappa_j\mu(j)$, and $f'(\Psi;\Theta)=\mathrm{Re}\langle G_{B_\Theta},\Theta\rangle$ for
some ordered natural orbital basis $B_\Theta$ of $\Psi$.
\item[(ii)] if $\kappa$ is non-decreasing on every block, then
$f'(\Psi;\Theta)\le\mathrm{Re}\langle G_B,\Theta\rangle$ for \emph{every} ordered natural orbital
basis $B$ of $\Psi$.
\end{enumerate}
\end{lemma}

\begin{proof}
\emph{The two vectors $G_B$ and $\gamma'$.} Since $\gamma'$ is built from
$\ket\Psi\bra\Theta+\ket\Theta\bra\Psi$ exactly as $\gamma^\Psi$ is from $\ket\Psi\bra\Psi$,
\eqref{eq:Tr} gives $\tr[h\gamma']=2\,\mathrm{Re}\langle\Theta,\Gamma(h)\Psi\rangle$ for
Hermitian $h$. As $\mathrm{Re}\langle\Theta,\Psi\rangle=0$, adding multiples of $\Psi$ to the
right-hand argument changes nothing, so with $h=K_B$,
\begin{equation}\label{eq:polar}
\tr[K_B\gamma']=2\,\mathrm{Re}\langle\Theta,\Dhat_B\Psi\rangle=\mathrm{Re}\langle G_B,\Theta\rangle .
\end{equation}
Also $\mathrm{Re}\langle\Psi,G_B\rangle=2\big(\langle\Psi,\Dhat_B\Psi\rangle-f(\Psi)\big)=0$ by
\eqref{eq:average}, so $G_B\in T_\Psi S$. It remains to compute $f'(\Psi;\Theta)$ and compare it
with $\tr[K_B\gamma']$.

\emph{(i).} By \eqref{eq:gamma}, $\gamma^{\Psi_s}=\gamma^\Psi+s\gamma'+O(s^2)$ with $\gamma'$
Hermitian, so by first-order perturbation theory \cite[Ch.~II, Thm.~6.8 and its proof]{Kato95} the eigenvalues of
$\gamma^{\Psi_s}$ emanating from a block $b$ are $\lambda^\Psi_j+s\mu(j)+o(s)$, $j\in b$,
occupying the positions of $b$ in that order. Hence $f'(\Psi;\Theta)=\sum_j\kappa_j\mu(j)$.
Assembling an eigenbasis of $P_b\gamma'P_b$ on each $W_b$ gives an ordered natural orbital basis
$B_\Theta$ whose diagonal for $\gamma'$ is $(\mu(j))_j$, so
$\tr[K_{B_\Theta}\gamma']=\sum_j\kappa_j\mu(j)$, and \eqref{eq:polar} finishes (i).

\emph{(ii).} For any ordered natural orbital basis $B$, the diagonal of $P_b\gamma'P_b$ in
$(u_j)_{j\in b}$ is majorized by $(\mu(j))_{j\in b}$, and a non-decreasing $\kappa_b$ paired
with the non-increasing $\mu$ minimizes $\sum_{j\in b}\kappa_jx_j$ over all $x$ majorized by
$\mu$; see \cite[Exercise~II.1.12]{Bhatia97} for Schur's theorem. Summing over blocks, $f'(\Psi;\Theta)=\sum_j\kappa_j\mu(j)\le\tr[K_B\gamma']$, which is
$\mathrm{Re}\langle G_B,\Theta\rangle$ by \eqref{eq:polar}.
\end{proof}

For pinned states this already settles the case of non-decreasing $\kappa$, with no hypothesis
on $F_D$. The degeneracy is not resolved---$f$ is still only one-sided differentiable, and a
block still carries many bases---but each of them then overestimates the derivative, so the fact
that $f$ is minimal at a pinned state forces $\Dhat_B\Psi=0$ for all of them at once.

\begin{corollary}[Monotone blocks]\label{cor:mono}
Let $D$ be valid and $\Psi$ pinned to $D$. If $\kappa$ is non-decreasing on every block of
$\Psi$, then $\Dhat_B\Psi=0$ for every ordered natural orbital basis $B$ of $\Psi$.
\end{corollary}

\begin{proof}
As $f\ge0=f(\Psi)$, the difference quotients defining $f'(\Psi;\Theta)$ are non-negative, so
$f'(\Psi;\Theta)\ge0$ and hence $\mathrm{Re}\langle G_B,\Theta\rangle\ge0$ for every unit
$\Theta\in T_\Psi S$, by Lemma~\ref{lem:dirderiv}(ii). Since $G_B\in T_\Psi S$, taking
$\Theta=-G_B/\|G_B\|$ gives $G_B=0$, that is $\Dhat_B\Psi=f(\Psi)\Psi=0$.
\end{proof}

When $\kappa$ decreases somewhere on a block, the basis matters (Examples~\ref{ex:A},
\ref{ex:B}), and the slope of $f$ can be controlled only after averaging over the bases of that
block. That averaging is the smearing of $\mathcal K(\Psi)$ below, and its cost is the constant
$\beta_D$, which we define next.

\subsection*{The constant}

What survives the averaging is the average of $D$ over the configurations the state occupies,
weighted by $|c_I|^2$, and the constant is the smallest positive value such an average can take.
The weights are not free: on a block the occupation numbers are equal, so, since $\lambda_j=\sum_{I\ni j}|c_I|^2$ in an ordered natural orbital basis,
each position of the block is occupied equally often by the configurations, counted with their
weights. Only the average matters, so one can work with integer weights.

Here is such a list for the Borland--Dennis constraint $D=\lambda_5+\lambda_6-\lambda_4$, five
configurations with integer weights $r_I$ making the orbitals $1,\dots,5$ each occupied four
times in all:
\[
\begin{array}{c|ccccc|c}
 & \{1,2,3\} & \{2,4,5\} & \{1,3,5\} & \{1,4,6\} & \{3,4,5\} & \\\hline
r_I & 2 & 2 & 1 & 1 & 1 & |r|:=\sum_Ir_I=7\\
D(\ind_I) & 0 & 0 & 1 & 0 & 0 & \textstyle\sum_Ir_ID(\ind_I)=1
\end{array}
\]
The average of $D$ over the list is $\tfrac17$: one unit of constraint value spread over seven.
The weights of a state whose occupation numbers are constant on $\{1,\dots,5\}$ obey the same
equal-occupation condition, for the same reason; balanced lists relax that condition to arbitrary
non-negative integers, and the smallest positive average over them is the constant for this
pattern. The relaxation is strict---the list above normalizes to
$(\tfrac47,\dots,\tfrac47,\tfrac17)$, which is not in $\Delta_{6,3}$---which is one reason
$\beta_D$ is not sharp (see Remark~\ref{rem:alpha} below).

Two features of the example drive the general definition. The equal occupation is imposed on
the interval $\{1,\dots,5\}$ of positions, and the longer such an interval is, the more
equalities it imposes, the further the constraint value must be spread, and the smaller the
average --- so degeneracy is what costs, and the constant must minimize over all the ways a
state can be degenerate. But only degeneracy along which $\kappa$ falls costs anything: where
$\kappa$ is non-decreasing on a block, Corollary~\ref{cor:mono} gives the rule outright and no
averaging is needed. Here $\kappa=(0,0,0,-1,1,1)$ falls at the fourth position, inside
$\{1,\dots,5\}$.

Accordingly, a \emph{pattern} is a partition of $\{1,\dots,d\}$ into intervals, recording which
positions a state may have degenerate, and it is \emph{admissible} if $\kappa$ is not
non-decreasing on any of its intervals of length $\ge2$---that is, if along each such interval
$\kappa$ drops at some step, $\kappa_j>\kappa_{j+1}$ (the non-decreasing case is already settled by Corollary~\ref{cor:mono}). Fix a pattern and let $r$ be
a non-zero vector of non-negative integers $r_I$, one for each configuration. Call $r$
\emph{balanced} if within each interval every position is occupied equally often,
$\sum_{I\ni j}r_I=\sum_{I\ni j'}r_I$ for $j,j'$ in the same interval, and \emph{minimal} if no
$r$ supported on a proper subset of its support is balanced. In the example every position of
$\{1,\dots,5\}$ is occupied four times, so $r$ is balanced; and it is minimal, since the four balance equations on these five configurations have a one-dimensional solution space, spanned by $r$, so a balanced vector supported on fewer of them would be a multiple of $r$, which is supported on all five.

A minimal balanced $r$ spans an extreme ray of the rational cone cut out by the balance
conditions, so it is determined by its support up to a positive scalar. Normalize it to the
integer vector whose entries have no common factor; this fixes $|r|$, which \eqref{eq:rho} below
bounds, while the average $\langle D\rangle_r:=|r|^{-1}\sum_Ir_ID(\ind_I)$ does not depend on
the choice. Set
\begin{equation}\label{eq:beta}
\beta_D:=\min\big\{\langle D\rangle_r\ \big|\ r\ \text{minimal balanced for an admissible
pattern},\ \langle D\rangle_r>0\big\}.
\end{equation}
The example gives $\beta_D\le\tfrac17$ for Borland--Dennis; equality holds, by the enumeration
described in the last remark of this appendix. The set in \eqref{eq:beta} is non-empty unless $f\equiv0$: decomposing the weights of
a state with $f>0$ as in the proof of Lemma~\ref{lem:beta} produces a minimal balanced $r$ with
positive average. It is finite, so $\beta_D>0$. The pattern with singleton intervals is
admissible and its minimal balanced vectors are the single configurations, so $\beta_D\le a$,
the smallest positive value of $D(\ind_I)$.

For $D$ with integer coefficients, which every generalized Pauli constraint has up to scaling
\cite{AltunbulakKlyachko08} (and scaling $D$ leaves $f/\beta_D$ unchanged), one has
\begin{equation}\label{eq:rho}
\beta_D\ \ge\ \frac{1}{d\,h(d)}\ \ge\ \frac{2^d}{d\,(d+1)^{(d+1)/2}},
\end{equation}
where $h(n)$ is the largest determinant of an $n\times n$ matrix with entries $0,1$. For
Borland--Dennis this gives $\tfrac1{54}$, since $h(6)=9$, and about $\tfrac1{85}$ for the closed-form estimate. To
see \eqref{eq:rho}, write the balance conditions with unknowns $(r,c)\ge0$ as
$\sum_Ir_I\ind_I(j)=c_b$ for $j$ in the interval $b$. Up to the sign of the $c$-columns the
coefficient matrix has entries $0,1$ and at most $d$ rows, so an extreme ray of the cone these
equations cut out of $\{(r,c)\ge0\}$ has support of size at most $d+1$, at most $d$ of it in the
$r$-variables since some $c_b$ is positive, and by Cramer's rule a primitive integer
representative whose entries are minors of order $\le d$. Hence $|r|\le d\,h(d)$, while
$|r|\langle D\rangle_r$ is a positive integer when it is positive. Finally
$h(n)\le2^{-n}(n+1)^{(n+1)/2}$: bordering an $n\times n$ matrix $A$ with entries $0,1$ by a
first row and column of ones and replacing $A$ by $J-2A$, $J$ the all-ones matrix, gives an
$(n+1)\times(n+1)$ matrix with entries $\pm1$ and determinant $\pm2^n\det A$, to which
Hadamard's inequality applies.

With $\beta_D$ in hand, the distance bound reads as follows.

\begin{theorem}\label{thm:exponent}
Let $D$ be valid, with $F_D\ne\emptyset$ and $f\not\equiv0$. For every normalized $\Psi$ with
$f(\Psi)<\beta_D$,
\[
\dist(\Psi,R_D)\ \le\ \arcsin\sqrt{f(\Psi)/\beta_D}.
\]
In particular $\dist(\Psi,R_D)\le\sqrt{2f(\Psi)/\beta_D}$ when $f(\Psi)\le\beta_D/2$.
\end{theorem}

For a given state, most patterns cannot occur nearby, and the constant improves accordingly:
nothing in \eqref{eq:beta} need be minimized over degeneracies that no state within reach can
develop. Fix a radius $\rho>0$ and call a pattern \emph{compatible} with $(\Psi,\rho)$ if
$\lambda^\Psi_j-\lambda^\Psi_k\le2\rho$ for the endpoints $j<k$ of each of its intervals of
length $\ge2$, and let $\beta_D(\Psi,\rho)\ge\beta_D$ be the minimum \eqref{eq:beta} over
admissible patterns compatible with $(\Psi,\rho)$. Only compatible patterns arise within distance $\rho$ of $\Psi$: moving that far moves each
occupation number by at most $\rho$, so two of them can meet only if they start within $2\rho$
of each other.

\begin{proposition}[Local version]\label{prop:local}
Let $\Psi$ be normalized, $\rho>0$, $\beta:=\beta_D(\Psi,\rho)$ with $f(\Psi)<\beta$, and
$\eta:=\arcsin\sqrt{f(\Psi)/\beta}$. If $\eta<\rho$, then $\dist(\Psi,R_D)\le\eta$.
\end{proposition}

\subsection*{Smearing}

It remains to make the averaging precise and to bound the slope of $f$ by it. Fix a normalized
$\Psi$ and an ordered natural orbital basis $B_0$ of $\Psi$. Call a block $b$ of $\Psi$
\emph{decreasing} if $\kappa$ drops somewhere along it, $\kappa_j>\kappa_{j+1}$ for some
consecutive pair in $b$---the condition that makes an interval admissible---and let $W_b$ be
the corresponding eigenspace of $\gamma^\Psi$. By Lemma~\ref{lem:dirderiv}(ii) only these blocks
need averaging, so let the \emph{smearing set} $\mathcal K(\Psi)$ consist of the Hermitian
operators that agree with $K_{B_0}$ off the decreasing blocks and on each decreasing block $b$
act on $W_b$ with spectrum majorized by $\kappa_b$. Two properties are wanted. It is convex,
since majorization constrains the sums of the $k$ largest eigenvalues, each a convex function of a Hermitian matrix (Ky Fan, \cite[Exercise~II.1.13]{Bhatia97}). And it contains $K_B$ for every ordered natural
orbital basis $B$ agreeing with $B_0$ off the decreasing blocks, since such a $K_B$ acts on $W_b$ with spectrum $\kappa_b$.

Each $\Kbar\in\mathcal K(\Psi)$ is itself diagonal in one ordered natural orbital basis
$\bar B=(\bar u_j)$: take $B_0$ off the decreasing blocks and an eigenbasis of $\Kbar$ on each
$W_b$. Write $\Kbar=\sum_j\kbar_j\ket{\bar u_j}\bra{\bar u_j}$, so that $\kbar_j=\kappa_j$ off
the decreasing blocks, while on each of them $\kbar_b$ is majorized by $\kappa_b$ and in
particular $\sum_{j\in b}\kbar_j=\sum_{j\in b}\kappa_j$. With
$\Dbar(x):=\kappa_0+\sum_j\kbar_jx_j$ and $\Psi=\sum_I\bar c_I\ket{\bar u_I}$, the operator
$\Dhat_{\Kbar}$ is diagonal in the configurations of $\bar B$ with eigenvalues $\Dbar(\ind_I)$,
whence
\begin{equation}\label{eq:smearnorm}
\|\Dhat_{\Kbar}\Psi\|^2=\sum_I|\bar c_I|^2\,\Dbar(\ind_I)^2,\qquad
\langle\Psi,\Dhat_{\Kbar}\Psi\rangle=\sum_I|\bar c_I|^2\,\Dbar(\ind_I)=f(\Psi) .
\end{equation}
The second identity is the one that matters: smearing changes the individual values
$\Dbar(\ind_I)$ but not their average, because $\lambda^\Psi_j=\sum_{I\ni j}|\bar c_I|^2$ is
constant on blocks, so $\sum_j(\kbar_j-\kappa_j)\lambda^\Psi_j=0$, and \eqref{eq:average}
applies.

Since $f$ has no derivative at a degenerate state, its descent is measured by the \emph{strong
slope} of a continuous $h$ at $\Psi$,
\[
|\partial h|(\Psi):=\limsup_{\Phi\to\Psi}\big(h(\Psi)-h(\Phi)\big)_+/\|\Phi-\Psi\| ,
\]
the fastest rate at which $h$ decreases near $\Psi$. The next lemma is the point of the
smearing: it replaces the direction-dependent basis of Lemma~\ref{lem:dirderiv}(i) by the single
convex set $\mathcal K(\Psi)$, over which a worst case can be taken.

\begin{lemma}\label{lem:smear}
For every normalized $\Psi$,\quad
$|\partial f|(\Psi)^2\ \ge\
4\Big(\min_{\Kbar\in\mathcal K(\Psi)}\|\Dhat_{\Kbar}\Psi\|^2-f(\Psi)^2\Big)$.
\end{lemma}

\begin{proof}
Every direction gives a lower bound for the slope. Indeed, for a unit $\Theta\in T_\Psi S$ let
$B'_\Theta$ agree with the basis $B_\Theta$ of Lemma~\ref{lem:dirderiv}(i) on the decreasing
blocks and with $B_0$ elsewhere; then $K_{B'_\Theta}\in\mathcal K(\Psi)$, and
Lemma~\ref{lem:dirderiv}(ii), applied on the non-decreasing blocks, gives
$f'(\Psi;\Theta)\le\mathrm{Re}\langle G_{B'_\Theta},\Theta\rangle$. Since
$\|\Psi_s-\Psi\|=s+O(s^3)$, the quotients defining $|\partial f|(\Psi)$ along $\Psi_s$ tend to
$-f'(\Psi;\Theta)$, so
\[
|\partial f|(\Psi)\ \ge\ -f'(\Psi;\Theta)\ \ge\ -\mathrm{Re}\langle G_{B'_\Theta},\Theta\rangle .
\]

It remains to choose $\Theta$ well. Put
$C:=\{2(\Dhat_{\Kbar}\Psi-f\Psi)\,|\,\Kbar\in\mathcal K(\Psi)\}$, a convex compact subset of
$T_\Psi S$, being the affine image of $\mathcal K(\Psi)$; note that $G_{B'_\Theta}\in C$ for
every $\Theta$. If $0\in C$ then $\Dhat_{\Kbar}\Psi=f\Psi$ for some $\Kbar$, so
$\min_{\Kbar}\|\Dhat_{\Kbar}\Psi\|^2\le f^2$ and the claim is trivial. Otherwise let $G^*$ be
the point of $C$ nearest to $0$, so that $\mathrm{Re}\langle G,G^*\rangle\ge\|G^*\|^2$ for all
$G\in C$, and take $\Theta^*:=-G^*/\|G^*\|$. Then
$-\mathrm{Re}\langle G,\Theta^*\rangle\ge\|G^*\|$ for all $G\in C$, in particular for
$G=G_{B'_{\Theta^*}}$, so by the display above
\[
|\partial f|(\Psi)\ \ge\ \|G^*\|=\dist(0,C)=2\min_{\Kbar}\|\Dhat_{\Kbar}\Psi-f\Psi\| .
\]
Squaring, and substituting $\|\Dhat_{\Kbar}\Psi-f\Psi\|^2=\|\Dhat_{\Kbar}\Psi\|^2-f^2$ from
\eqref{eq:smearnorm}, gives the lemma.
\end{proof}

The second lemma bounds the right-hand side of Lemma~\ref{lem:smear} from below, and is where
$\beta_D$ enters: the weights of $\Psi$ are decomposed into the balanced vectors of the previous
subsection, and each of those is controlled by its average of $D$.

\begin{lemma}\label{lem:beta}
Let $\Psi$ be normalized, let $\mathcal S$ be the pattern whose intervals of length $\ge2$ are
the decreasing blocks of $\Psi$, and let $\beta>0$ be such that $\langle D\rangle_r\ge\beta$ for
every minimal balanced $r$ for $\mathcal S$ with $\langle D\rangle_r>0$. Then
$\|\Dhat_{\Kbar}\Psi\|^2\ge\beta f(\Psi)$ for every $\Kbar\in\mathcal K(\Psi)$.
\end{lemma}

\begin{proof}
Take $\bar B$, $\kbar$, $\bar c$ as in \eqref{eq:smearnorm}, and let $Q$ be the cone of $w\ge0$,
indexed by configurations, with $\sum_{I\ni j}w_I$ constant along each decreasing block. The
weights $w:=(|\bar c_I|^2)_I$ of $\Psi$ lie in $Q$, since $\lambda^\Psi$ is constant on blocks.
As $Q$ is a pointed polyhedral cone, its extreme rays are its elements of minimal support, that
is, the minimal balanced vectors for $\mathcal S$, and $w=\sum_kt_kr_k$ with $t_k\ge0$ and $r_k$
minimal balanced (Minkowski--Weyl).

Put $F(v):=\sum_Iv_ID(\ind_I)$ for $v\in Q$. Smearing does not change $F$: as $\kbar-\kappa$
sums to zero along each decreasing block, while $\sum_{I\ni j}v_I$ is constant there,
$\sum_Iv_I\Dbar(\ind_I)=F(v)$. In particular $F(w)=f(\Psi)$ by \eqref{eq:average}.

Now bound each ray separately. If $F(r_k)>0$ then $\langle D\rangle_{r_k}=F(r_k)/|r_k|\ge\beta$,
so by Cauchy--Schwarz
\[
\sum_Ir_{k,I}\Dbar(\ind_I)^2\ \ge\ \frac{\big(\sum_Ir_{k,I}\Dbar(\ind_I)\big)^2}{|r_k|}
=\frac{F(r_k)^2}{|r_k|}\ \ge\ \beta F(r_k) ,
\]
while for the other $k$ the left-hand side is still non-negative. Summing with the weights
$t_k$, and then dropping the rays with $F(r_k)\le0$, which only increases the right-hand side,
\[
\|\Dhat_{\Kbar}\Psi\|^2=\sum_Iw_I\Dbar(\ind_I)^2\ \ge\ \beta\sum_{k:\,F(r_k)>0}t_kF(r_k)\ \ge\
\beta F(w)=\beta f(\Psi).\qedhere
\]
\end{proof}

Together the two lemmas bound the slope of $f$ from below wherever $0<f<\beta$, and Ekeland's
principle converts a slope bound into a distance. The theorem and the proposition are proved
together, differing only in the space on which the principle is applied.

\begin{proof}[Proof of Theorem~\ref{thm:exponent} and Proposition~\ref{prop:local}]
Let $\beta$ be $\beta_D$, or $\beta_D(\Psi,\rho)$ in the local case. Let $Y$ be $S$, or the closed ball of radius $\rho$ about $\Psi$ in $S$; both are closed in
$\wedge^N\C^d$, hence complete, as Ekeland's principle \cite{Ekeland74} requires.

\emph{The hypothesis of Lemma~\ref{lem:beta} holds with this $\beta$ at every $y\in Y$.} Its
pattern $\mathcal S$ is admissible, since its intervals of length $\ge2$ are decreasing blocks.
In the local case it is also compatible with $(\Psi,\rho)$: by Step~3 of the proof of
Theorem~\ref{thm:quasi}, $\|\gamma^y-\gamma^\Psi\|\le\|y-\Psi\|\le\rho$, so Weyl's inequality
moves each occupation number by at most $\rho$, and two of them can coincide at $y$ only if they
differ by at most $2\rho$ at $\Psi$.

\emph{The slope bound.} By Lemmas~\ref{lem:smear} and~\ref{lem:beta},
$|\partial f|(y)\ge2\sqrt{\beta f(y)-f(y)^2}$ whenever $y\in Y$ and $0<f(y)<\beta$. Put
$h:=\varphi\circ\min(f,\beta)$ with $\varphi(t):=\arcsin\sqrt{t/\beta}$, so that
$\varphi'(t)=1/\big(2\sqrt{t(\beta-t)}\big)$ and the chain rule for the strong slope gives
$|\partial h|(y)\ge1$ at every such $y$: the two factors are exact reciprocals.

\emph{Ekeland.} Fix $\sigma\in(0,1)$, with $h(\Psi)/\sigma<\rho$ in the local case. Ekeland's
principle \cite{Ekeland74} on $Y$, applied to $h$ with $\varepsilon:=h(\Psi)$ and
$\delta:=\varepsilon/\sigma$---legitimate since $h\ge0$, so $h(\Psi)\le\inf_Yh+\varepsilon$---yields $y\in Y$ with
\[
h(y)\le h(\Psi),\qquad \|y-\Psi\|\le h(\Psi)/\sigma,\qquad
h(z)>h(y)-\sigma\|z-y\| \ \text{ for } z\in Y\setminus\{y\} .
\]
The third property says $|\partial h|(y)\le\sigma<1$; the first gives $f(y)<\beta$; and by the
second, in the local case, $\|y-\Psi\|<\rho$, so $y$ is interior to $Y$ and the slope bound
applies there. Were $f(y)>0$ we would therefore have $|\partial h|(y)\ge1$, a contradiction; so
$f(y)=0$, that is $y\in R_D$, and $\dist(\Psi,R_D)\le\|y-\Psi\|\le h(\Psi)/\sigma$. Now let $\sigma\uparrow1$.
\end{proof}

\paragraph{Remarks.}
\begin{enumerate}
\item\label{rem:alpha} The constant $\beta_D$ is not the best the method allows. Replacing it by
$\alpha_D:=\inf\|\Dhat_{\Kbar}\Psi\|^2/f(\Psi)$, the infimum over normalized $\Psi$ with
$f(\Psi)>0$, bases $B_0$ and $\Kbar\in\mathcal K(\Psi)$, leaves the proof intact, and
$\alpha_D\ge\beta_D$ by Lemma~\ref{lem:beta}. Computing $\alpha_D$ needs a non-convex
minimization over the smearings on top of the enumeration of Remark~\ref{rem:computing}.

\item\label{rem:sharp} The exponent $\tfrac12$ cannot be improved, already at a Slater
determinant. For $D=1-\lambda_1$ on $\wedge^4\C^8$ and
$\Psi_t\propto\ket{u_1\wedge\dots\wedge u_4}+t\ket{u_5\wedge\dots\wedge u_8}$ one has
$f(\Psi_t)=t^2/(1+t^2)$, while every pinned $\Phi$ has a fully occupied orbital $u$, so
$|\langle\Phi,\Psi_t\rangle|^2\le\langle\Psi_t,\nhat_u\Psi_t\rangle\le\lambda_1(\Psi_t)$ and
the distance from $\Psi_t$ to the pinned set satisfies
$\dist(\Psi_t,R_D)^2\ge2-2\sqrt{\lambda_1(\Psi_t)}=t^2+O(t^4)$. Thus $\dist(\Psi_t,R_D)$ and
$\sqrt{f(\Psi_t)}$ are of the same order as $t\to0$.

\item\label{rem:local} Near a pinned state with non-degenerate occupation numbers, or with
$\kappa$ non-decreasing on each of its blocks, no block of a nearby state is decreasing. There
$\mathcal K(\Psi)=\{K_B\}$ and $\|\Dhat_B\Psi\|^2\ge af$, since $x^2\ge ax$ for every value
$x=D(\ind_I)$ \cite{SchillingBenavidesVrana17}. Applying Proposition~\ref{prop:local} on a ball
small enough to contain no other degeneracies then gives $\arcsin\sqrt{f/a}$, smaller for small $f$ than the distance $\sqrt{2f}$ obtained in \cite{SchillingBenavidesVrana17}.

\item\label{rem:spin} Everything transfers to $\mathcal H^{S,M}_N$, $\gamma_l^\Psi$ and
$\Delta^S_{d,N}$, with spatial occupation vectors in $\{0,1,2\}^d$ in place of $\ind_I$. Three things change. In \eqref{eq:rho} the balance matrix then has entries $0,1,2$, and since the determinant
is multilinear in its rows this replaces $h(d)$ by $2^dh(d)$, so that
$\beta_D\ge1/\big(d\,(d+1)^{(d+1)/2}\big)$. Since a spatial occupation number operator has eigenvalues $0,1,2$, moving $\Psi$ by $\rho$ moves each spatial occupation number by up to $2\rho$, so the compatibility window $2\rho$ becomes $4\rho$; and the bound on $E$ becomes $\|E\|\le2\eta$, as in Section~\ref{sec:main}.

\item\label{rem:computing} \emph{Computing $\beta_D$.} Balance sees a configuration only through
its occupation of the intervals of the pattern, so configurations may be grouped accordingly;
the extreme rays of the resulting cone are then enumerated exactly (double description), and for
each ray one picks the representative configurations minimizing the positive average. For
Borland--Dennis this gives $\beta_D=\tfrac17$. For beryllium (Example~\ref{ex:Be}) take
$\rho=10^{-2}$: a compatible interval lies inside $\{1,\dots,4\}$ or inside $\{5,\dots,10\}$,
and $\kappa=(-2,-1,-1,0)$ is non-decreasing on the first, so an admissible compatible pattern
has a single interval $b\subseteq\{5,\dots,10\}$ containing $7$ and $8$. The nine such intervals
give $\beta=\tfrac12$ for $b=\{7,8\}$; $\tfrac13$ for $\{6,7,8\}$ and $\{7,8,9\}$; $\tfrac15$
for $\{5,\dots,8\}$, $\{6,\dots,9\}$ and $\{7,\dots,10\}$; $\tfrac18$ for $\{5,\dots,9\}$;
$\tfrac17$ for $\{6,\dots,10\}$; and $\tfrac1{11}$ for $\{5,\dots,10\}$. Hence $\beta_D(\Psi,10^{-2})=\tfrac1{11}$, and
$\arcsin\sqrt{11f}=7.0\cdot10^{-4}<\rho$, as Proposition~\ref{prop:local} requires.
\end{enumerate}

\end{document}